%% file: ITW17_submitted.tex
\documentclass[10pt,journal]{IEEEtran}

\usepackage{cite}
\ifCLASSINFOpdf
   \usepackage[pdftex]{graphicx}
\else
  \usepackage[dvips]{graphicx}
\fi
\usepackage{amsmath,amssymb,amsthm,mathrsfs,amsfonts,dsfont,stmaryrd}
\usepackage{array}
\usepackage{dsfont}
\usepackage{amsbsy}
\usepackage{epstopdf}
\usepackage{graphicx,color}

\usepackage[lowtilde]{url}

\newtheorem{lemma}{Lemma}
\newtheorem{proposition}{Proposition}

\long\def\symbolfootnote[#1]#2{\begingroup%
\def\thefootnote{\fnsymbol{footnote}}\footnote[#1]{#2}\endgroup}
\newtheorem{theorem}{Theorem}

\newcommand{\dv}{\mathbf} 
\newcommand{\mc}{\mathcal} 
\newcommand{\mkv}{-\!\!\!\!\minuso\!\!\!\!-}

 \allowdisplaybreaks

\newcommand*{\str}{^{\mathsf{*}}}

\newcommand{\abs}[1]{\lvert#1\rvert}

\DeclareMathOperator*{\argmin}{min}

\newcounter{MYtempeqncnt}

\usepackage{dblfloatfix}

\usepackage{algorithm} 
\usepackage{algpseudocode}

\algnewcommand{\Inputs}[1]{%
	\State \textbf{input:}
	\parbox[t]{.8\linewidth}{\raggedright #1}
}
\algnewcommand{\Initialize}[1]{%
	\State \textbf{initialization}
	\parbox[t]{.95\linewidth}{\raggedright #1}
}

\algnewcommand{\Outputs}[1]{%
	\State \textbf{output:}
	\parbox[t]{.8\linewidth}{\raggedright #1}
}

\usepackage{tabularx}

\usepackage{tikz}
\usepackage{pgfplots}
\usetikzlibrary{calc} 
\usetikzlibrary{patterns} 
\usetikzlibrary{shapes.geometric}

\tikzstyle{block}=[draw, rectangle, text centered, minimum width=2em, minimum height=3em]
\tikzstyle{sum}=[draw, circle, minimum size=.2cm]
\tikzstyle{rect}=[draw, dashed, rectangle, text centered, minimum width=12em, minimum height=8em]
\tikzstyle{oval}=[draw, ellipse, minimum size=1em, minimum height=1em, align=center]

\begin{document}
\fontencoding{OT1}\fontsize{10}{11}\selectfont

\pagenumbering{gobble}

\title{A Generalization of Blahut-Arimoto Algorithm to Compute Rate-Distortion Regions of Multiterminal Source Coding Under Logarithmic Loss} 
\author{
Yi{\u{g}}it U{\u{g}}ur $^{\dagger}$$^{\ddagger}$ \qquad \quad I\~naki Estella Aguerri $^{\dagger}$ \qquad \quad Abdellatif Zaidi $^{\dagger}$$^{\ddagger}$ \vspace{0.1cm} \\   
{\small
$^{\dagger}$ Mathematics and Algorithmic Sciences Lab., France Research Center, Huawei Technologies, Boulogne-Billancourt, 92100, France\\
$^{\ddagger}$ Universit\'e Paris-Est, Champs-sur-Marne, 77454, France\\
\{\tt yigit.ugur@huawei.com,  inaki.estella@huawei.com, abdellatif.zaidi@u-pem.fr\} }
\vspace{-8mm}

} 

\maketitle
\begin{abstract} 
In this paper, we present iterative algorithms that numerically compute the rate-distortion regions of two problems: the two-encoder multiterminal source coding problem and the Chief Executive Officer (CEO) problem, both under logarithmic loss distortion measure. With the clear connection of these models with the distributed information bottleneck method, the proposed algorithms may find usefulness in a variety of applications, such as clustering, pattern recognition and learning. We illustrate the efficiency of our algorithms through some numerical examples. 
\end{abstract}

\IEEEpeerreviewmaketitle

\vspace{-2mm}
\section{Introduction}

The logarithmic loss (log-loss) function is a widely used penalty function that is particularly natural in settings in which reconstructions are allowed to be `soft', rather than `hard' or deterministic. That is, settings in which decoders or estimators output not only estimate values but also assessment of the levels of confidence in those values. More specifically, for a length-$n$ vector or sequence $\dv x=(x_1,\hdots,x_n)$ with element $x_i$, $i=1,\hdots,n$, in some alphabet $\mc X_i$, its reconstruction version or estimate is a vector $\hat{\dv x}=(\hat{x}_1,\hdots,\hat{x}_n)$ for which every component $\hat{x}_i$ is a probability distribution on $\mc X_i$. The symbol-wise distortion between $x_i$ and $\hat{x}_i$ is measured as
\begin{equation}
\vspace{-1mm}
d(x_i,\hat{x}_i) = \log\Big(\frac{1}{\hat{x}_i(x_i)}\Big),
\label{definition-log-loss-distortion-measure}
\end{equation}
where $\hat{x}_i(x_i)$ represents the value of the probability distribution $\hat{x}_i$ evaluated for the outcome $x_i$. Using this symbol-wise distortion, distortion between sequences is then defined as
\begin{equation*}
d^{(n)}(\dv x, \hat{\dv x}) = \frac{1}{n} \sum_{i=1}^{n} d(x_i,\hat{x}_i).
\vspace{-1mm}
\end{equation*}

The logarithmic loss function~\eqref{definition-log-loss-distortion-measure} has many appreciable features. First, it is used as a penalty criterion in various contexts, including clustering and classification~\cite{Tishby}, pattern recognition, learning and prediction~\cite{C-BL06}, image processing~\cite{AABG06} and others. Second, it was recently shown in a remarkable paper by Courtade and Weissman~\cite{Weissman} to admit key properties that allow to solve multiterminal source coding problems that are known to be difficult otherwise, in the sense that their solutions are still to be found for general distortion measures. Specifically, as mentioned in~\cite{Weissman}, the log-loss distortion measure admits a lower bound in the form of conditional entropy. Using this key finding, Courtade and Weissman successfully establish the single-letter characterization of the achievable rate-distortion (RD) region of the classical two-encoder multiterminal source coding problem~\cite[Theorem 6]{Weissman}, as well as that of the Chief Executive Officer (CEO) problem~\cite[Theorem 3]{Weissman}, both under log-loss distortion measure.

\begin{figure}[h!]
\centering 
\small
\begin{tikzpicture}
\node (in1) at (0,0.8) [left]{$X$};
\draw (in1) -- (0.2,0.8);
\node[dotted,thick] (ch1) at (1.5,1.6) [oval, align=center] {Channel\\$p(y_{1}|x)$};	
\node[dotted,thick] (ch2) at (1.5,0) [oval, align=center] {Channel\\$p(y_{2}|x)$};
\draw[->] (0.2,0.8) |- (ch1.west);
\draw[->] (0.2,0.8) |- (ch2.west);
\node (enc1) at (4.5,1.6) [block, minimum width=2em, minimum height=2em, align=center] {Encoder 1};	
	\node (enc2) at (4.5,0) [block, minimum width=2em, minimum height=2em, align=center] {Encoder 2};	
	\draw[->] (ch1.east) to node[above] {$Y_{1}$} (enc1.west);
	\draw[->] (ch2.east) to node[below] {$Y_{2}$} (enc2.west);
	\node (dec) at (6.5,0.8) [block, minimum width=2em, minimum height=8em, align=center] {\rotatebox{90}{Decoder}};	
	\draw[->] (enc1.east) to node[above] {$R_{1}$} ($(dec.west)+(0,0.8)$);
	\draw[->] (enc2.east) to node[below] {$R_{2}$} ($(dec.west)+(0,-0.8)$);
	\draw[->] (dec.east) to ($(dec.east)+(0.3,0)$) node[right] {$\hat{X}$};	
	\end{tikzpicture}
	\vspace{-1mm}
	\caption{Chief Executive Officer (CEO) source coding problem.} 
	\label{fig:CEO}
	\vspace{-3mm}
\end{figure}
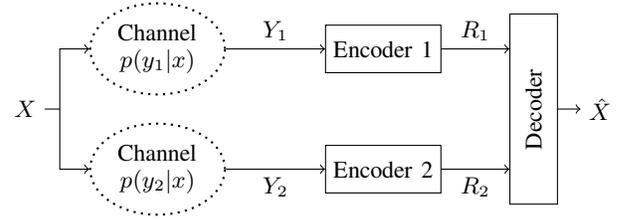

The computation of the RD regions of the aforementioned multiterminal source coding problems for general memoryless sources is important per-se; and even more considering the wide range of applications of lossy multiterminal source coding, including emerging applications in fields such as distributed learning and estimation~\cite[Chapter 9]{C-BL06}. For example, the information bottleneck method~\cite{Tishby} is an efficient data clustering algorithm, which essentially computes the RD region of a point-to-point rate-distortion problem, in which the distortion is measured under log-loss. Developing algorithms that allow to compute the RD region of multiterminal source coding problems can lead to efficient distributed algorithms for clustering and prediction.

Nonetheless, computing the RD region of multiterminal source-coding problems under log-loss for general memoryless sources is a difficult task, as it involves non-trivial optimization problems over distributions of auxiliary random variables. In this paper, we develop computational techniques for solving numerically the RD regions of  the two-encoder multiterminal source coding problem and the CEO problem, both under logarithmic loss distortion measure. Our approach for the computation of both regions consists on first reexpressing the original RD region in terms of the union of simpler regions, whose boundary points can be expressed parametrically. Then, each boundary point can be computed numerically via an appropriate iterative minimization method that we develop here. The proposed method can be regarded as a generalization of the well known Blahut-Arimoto (BA) algorithm~\cite{Blahut,Arimoto} to the aforementioned multiterminal settings. For other generalizations of this algorithm, the reader may refer to related works on point-to-point~\cite{CSX05,CB04} and broadcast and multiple access multiterminal settings~\cite{WeiYu,RG04}.

\begin{figure*}[!b]
	\hrulefill
	\small
	\setcounter{MYtempeqncnt}{\value{equation}}
	\setcounter{equation}{8}
	\begin{align}
	F_{\dv s}(\dv P) \triangleq& \: H(X|U_1,U_2) + s_1 I(Y_1;U_1|U_2)+ s_2 I(Y_2;U_2) = H(X|U_{1},U_{2}) + s_{1} [I(U_{1};Y_{1})-I(U_{1};U_{2})] + s_{2} I(U_{2};Y_{2}) \nonumber \\
	=&- \sum\nolimits_{u_{1}u_{2}x} p(u_{1},u_{2},x) \log p(x|u_{1},u_{2}) - s_{1} \sum\nolimits_{u_{1}u_{2}} p(u_{1},u_{2}) \log p(u_{1},u_{2}) - s_{2} \sum\nolimits_{u_{2}} p(u_{2}) \log p(u_{2}) \nonumber\\
	&+ s_{1} \sum\nolimits_{u_{1}y_{1}} p(u_{1}|y_{1}) p(y_{1}) \log p(u_{1}|y_{1}) + s_{2}
	\sum\nolimits_{u_{2}y_{2}} p(u_{2}|y_{2}) p(y_{2}) \log p(u_{2}|y_{2}) + s_{1} \sum\nolimits_{u_{2}} p(u_{2}) \log p(u_{2}), \label{eq:CEO_lag1}
	\end{align}
	\setcounter{equation}{\value{MYtempeqncnt}}
	\small
	\vspace{-1.5em}
	\setcounter{MYtempeqncnt}{\value{equation}}
	\setcounter{equation}{10}
	\begin{align}
	F_{\dv s}(\dv P, \dv Q) \triangleq& - \sum\nolimits_{u_{1}u_{2}x} p(u_{1},u_{2},x) \log q(x|u_{1},u_{2}) - s_{1} \sum\nolimits_{u_{1}u_{2}} p(u_{1},u_{2}) \log q(u_{1},u_{2}) - s_{2} \sum\nolimits_{u_{2}} p(u_{2}) \log q(u_{2}) \nonumber\\
	&+ s_{1} \sum\nolimits_{u_{1}y_{1}} p(u_{1}|y_{1}) p(y_{1}) \log p(u_{1}|y_{1}) + s_{2} \sum\nolimits_{u_{2}y_{2}} p(u_{2}|y_{2}) p(y_{2}) \log p(u_{2}|y_{2}) + s_{1} \sum\nolimits_{u_{2}} p(u_{2}) \log p(u_{2}).  \label{eq:CEO_lag2}
	\end{align}
	\setcounter{equation}{\value{MYtempeqncnt}}
\end{figure*}

\vspace{-2mm}
\section{The CEO Problem}
\label{sec:CEO} 
Consider the discrete memoryless two-encoder CEO setup shown in Figure~\ref{fig:CEO}. In this setup, $X$ is a discrete memoryless remote source with elements in some alphabet $\mc X$, and $Y_1$ and $Y_2$ are correlated memoryless observations or sources with elements in sets $\mc Y_1$ and $\mc Y_2$, respectively. The joint probability mass function (pmf) of the triple $(X,Y_{1},Y_{2})$ is $P_{X,Y_{1},Y_{2}}$, which is assumed here to satisfy the Markov chain $Y_1 \mkv X \mkv Y_2$. The source $Y_1$ is observed at Encoder 1 and the source $Y_2$ is observed at Encoder 2. The encoders are connected to a decoder through error-free bit-pipes of capacities $R_1$ and $R_2$, respectively. The decoder wants to reproduce an estimate $\hat{X}$ of the remote source $X$ to within some prescribed fidelity level $D$ where the distortion is evaluated using the measure~\eqref{definition-log-loss-distortion-measure}. That is, $\mathbb{E}[d(X,\hat{X})] \leq D$ with $d(\cdot)$ given by~\eqref{definition-log-loss-distortion-measure}.
 
First, we recall the following theorem from~\cite[Theorem 3]{Weissman} which characterizes the RD region of the CEO problem under log-loss measure. We define $i^\mathsf{c} \triangleq i \pmod{2} + 1$.
\begin{theorem}{~\cite[Theorem 3]{Weissman}}
\label{theo:CEO-theorem}
The  tuple $(R_{1},R_{2},D)\!\in\! \mathcal{RD}_{\mathrm{CEO}}$ is achievable for the CEO problem under log-loss iff
\begin{align}
R_{i} & \geq I(U_{i};Y_{i}|U_{i^\mathsf{c}},Q), \quad\quad\quad \text{for }\: i=1, 2, \label{eq:ConstRi}\\
R_{1} + R_{2} & \geq I(U_{1},U_{2};Y_{1},Y_{2}|Q), \label{eq:ConstRsum}\\ 
D & \geq H(X|U_{1},U_{2},Q),
\end{align}
for some pmf 
$p(x) p(y_{1}|x) p(y_{2}|x) p(u_{1}|y_{1},q) p(u_{2}|y_{2},q) p(q)$,
where $\abs{\mathcal{U}_{1}} \leq \abs{\mathcal{Y}_{1}}$, $\abs{\mathcal{U}_{2}} \leq \abs{\mathcal{Y}_{2}}$, and $\abs{\mathcal{Q}} \leq 4$. 
\end{theorem}

\begin{figure*}[!b]
	\normalsize
	\hrulefill
	\setcounter{MYtempeqncnt}{\value{equation}}
	\setcounter{equation}{14}
	\begin{equation}	
	\label{eq:rho_ceo}	
	\begin{aligned}
	\rho_{1}(u_{1},y_{1}) &\triangleq \frac{1}{s_{1}} \sum_{u_{2}x} p(x|y_{1}) p(u_{2}|x) \log q(x|u_{1},u_{2}) + \sum_{u_{2}x} p(x|y_{1}) p(u_{2}|x) \log q(u_{1},u_{2}),  \\
	\rho_{2}(u_{2},y_{2}) &\triangleq \frac{1}{s_{2}} \sum_{u_{1}x} p(x|y_{2}) p(u_{1}|x) \log q(x|u_{1},u_{2}) + \frac{s_{1}}{s_{2}} \sum_{u_{1}x} p(x|y_{2}) p(u_{1}|x) \log q(u_{1},u_{2})
	+ \log q(u_{2}) - \frac{s_{1}}{s_{2}} \log p(u_{2}). 
	\end{aligned}
	\end{equation}	
	\setcounter{equation}{\value{MYtempeqncnt}}
	\vspace{-0.2em}
	\hrulefill
	\footnotesize
	\setcounter{MYtempeqncnt}{\value{equation}}
	\setcounter{equation}{15}
	\begin{align}
	\log p(u_{2}|y_{2}) =& \: \frac{1}{s_{2}} \sum_{u_{1}x} p(x|y_{2}) p(u_{1}|x) \log q(x|u_{1},u_{2}) + \frac{s_{1}}{s_{2}} \sum_{u_{1}x} p(x|y_{2}) p(u_{1}|x) \log q(u_{1},u_{2}) + \log q(u_{2}) - \frac{s_{1}}{s_{2}} \log p(u_{2}) + \frac{\lambda_{2}(y_{2})}{s_{2}p(y_{2})} - \frac{s_{1}+s_{2}}{s_{2}}, \nonumber \\
	\log p(u_{1}|y_{1}) =& \: \frac{1}{s_{1}} \sum\nolimits_{u_{2}x} p(x|y_{1}) p(u_{2}|x) \log q(x|u_{1},u_{2}) + \sum\nolimits_{u_{2}x} p(x|y_{1}) p(u_{2}|x) \log q(u_{1},u_{2}) + \frac{\lambda_{1}(y_{1})}{s_{1}p(y_{1})} - 1.  \label{eq:CEO-logp} 
	\end{align}
	\setcounter{equation}{\value{MYtempeqncnt}}
\end{figure*}

In this section, we develop a BA-type algorithm that allows to compute the convex region $\mathcal{RD}_{\mathrm{CEO}}$ for general memoryless sources. The outline of the proposed method is as follows. First, we rewrite the RD region $\mathcal{RD}_{\mathrm{CEO}}$ in terms of the union of two simpler regions in Proposition~\ref{prop:UnionCEO}. The tuples lying on the boundary of each region are parametrically given in Theorem~\ref{theo:CEO-param}. Then, the boundary points of each simpler region are computed numerically via an alternating minimization method derived in Section~\ref{sec:algo} and detailed in Algorithm~\ref{algo:CEO}. Finally, the original RD region is obtained as the convex hull of the union of the tuples obtained for the two simple regions.

Due to the space limitations, some proofs are omitted or only outlined. The detailed proofs are relegated to the full version of this work \cite{Yigit:ITW:Full}.

\vspace{-4mm}
\subsection{Equivalent Parametrization of $\mathcal{RD}_{\mathrm{CEO}}$}

Define the two regions $\mathcal{RD}_{\mathrm{CEO}}^{1}$ and  $\mathcal{RD}_{\mathrm{CEO}}^{2}$ for $i\!=\!1,2$ as
\begin{equation*}
\mathcal{RD}_{\mathrm{CEO}}^{i}=\{ (R_{1},R_{2},D) \::\: D\geq D_{\mathrm{CEO}}^{i}(R_{1},R_{2}) \},
\end{equation*}
with
\vspace{-2mm}
\begin{align}
&D_{\mathrm{CEO}}^{i}(R_{1}, R_{2}) \triangleq \argmin \; H(X|U_{1},U_{2}) \label{eq:RDProblem}\\
&\text{s.t.} \quad 
R_{i} \geq I(Y_{i};U_{i}|U_{i^\mathsf{c}}) \text{ and } R_{i^\mathsf{c}} \geq I(Y_{i^\mathsf{c}};U_{i^\mathsf{c}}),  \nonumber
\end{align}
and the minimization is over set of joint measures $P_{U_1,U_2,X,Y_1,Y_2}$ that satisfy $U_1 \mkv Y_1 \mkv X \mkv Y_2 \mkv U_2$.

As stated in the following proposition, the region $\mathcal{RD}_{\mathrm{CEO}}$ of Theorem~\ref{theo:CEO-theorem} coincides with the convex hull of the union of the two regions $\mathcal{RD}_{\mathrm{CEO}}^{1}$ and  $\mathcal{RD}_{\mathrm{CEO}}^{2}$. 

\begin{proposition}\label{prop:UnionCEO}
The region $\mathcal{RD}_{\mathrm{CEO}}$ is given by
\begin{align}
\mathcal{RD}_{\mathrm{CEO}} = \mathrm{conv} ( \mathcal{RD}_{\mathrm{CEO}}^{1} \cup \mathcal{RD}_{\mathrm{CEO}}^{2} ).\label{eq:RDChull}
\end{align}
\end{proposition}

\begin{proof}
The outline of the proof is as follows. Let $P_{U_1,U_2,X,Y_1,Y_2}$ and $P_Q$ be  such that $(R_1,R_2,D)\in\mathcal{RD}_{\mathrm{CEO}} $. The polytope defined by the rate constraints \eqref{eq:ConstRi}-\eqref{eq:ConstRsum}, denoted by $\mc V$, forms a contra-polymatroid with $2!$ extreme points (vertices) \cite{Weissman,Chen:2008:TIT}. Given a permutation $\pi$ on $\{1,2\}$, the tuple
\begin{align}
\tilde{R}_{\pi(1)} = I(Y_{\pi(1)};U_{\pi(1)}),\; \tilde{R}_{\pi(2)} = I(Y_{\pi(2)};U_{\pi(2)}|U_{\pi(1)}), \nonumber
\end{align}
defines an extreme point  of $\mc V$ for each permutation. As shown in \cite{Weissman}, for every extreme point $(\tilde{R}_1,\tilde{R}_2)$ of $\mc V$, the point $(\tilde{R}_1,\tilde{R}_2,D)$ is achieved by time-sharing two successive Wyner-Ziv (SWZ) strategies. The set of  achievable tuples with such SWZ scheme is characterized by the convex hull of  $\mathcal{RD}_{\mathrm{CEO}}^{\pi(1)}$. 
Convexifying the union of both regions as in~\eqref{eq:RDChull}, we obtain the original RD region $\mc R_{\mathrm{CEO}}$.
\end{proof}
The main advantage of Proposition~\ref{prop:UnionCEO} it that it reduces the computation of region $\mc {RD}_{\mathrm{CEO}}$  to the computation of the two regions $\mc {RD}_{\mathrm{CEO}}^{i}$, $i=1,2$, whose boundary can be efficiently parameterized, leading to an efficient computational method. In what follows, we concentrate on  $\mathcal{RD}_{\mathrm{CEO}}^{1}$. The computation of $\mathcal{RD}_{\mathrm{CEO}}^{2}$ follows similarly, and is omitted  for brevity. 
Next theorem provides a parameterization of the boundary tuples of the region  $\mathcal{RD}_{\mathrm{CEO}}^{1}$ in terms, each of them, of an optimization problem over the pmfs $\dv P \triangleq \{P_{U_1|Y_1},P_{U_2|Y_2}\}$.
\begin{theorem}
\label{theo:CEO-param}
For each $\dv s \triangleq [s_1, s_2]$, $s_1 > 0$, $s_2 > 0$, define a rate-distortion tuple $(R_{1,\dv s}, R_{2,\dv s} ,D_{\dv s})$ parametrically given by 
\begin{align}
&D_{\dv s} = -s_1R_{1,\dv s}-s_2R_{2, \dv s}+ \min_{\dv P}F_{\dv s}(\dv P), \label{eq:Dparam} \\ 
&R_{1, \dv s} = I(Y_1;U_1\str|U_2\str), \quad R_{2, \dv s} = I(Y_2;U_2\str), \label{eq:Rparam}
\end{align}
where $F_{\dv s}(\dv P)$ is given in \eqref{eq:CEO_lag1};  $\dv P\str$ are the conditional pmfs yielding the minimum in \eqref{eq:Dparam} and $U_1^*, U_2^*$ are the auxiliary variables induced by $\dv P^*$. Then, we have: 
\begin{enumerate}
\item Each value of  $\dv s$ leads to a tuple $(R_{1,\dv s}, R_{2,\dv s}, D_{\dv s})$ on the distortion-rate curve
$D_{\dv s} = D_{\mathrm{CEO}}^{1}(R_{1,\dv s}, R_{2,\dv s})$.

\item For every point on the distortion-rate curve, there is an $\dv s$ for which \eqref{eq:Dparam} and \eqref{eq:Rparam} hold.
\end{enumerate}
\end{theorem}

\begin{proof}
Suppose that $\dv P\str$ yields the minimum in \eqref{eq:Dparam}. For this $\dv P$ we have $ I(Y_1;U_1|U_2) = R_{1, \dv s}$ and $I(Y_2;U_2) = R_{2, \dv s}$. Then, 
\addtocounter{equation}{1} 
\begin{align}
D_{\dv s} &= -s_1R_{1,\dv s}-s_2R_{2, \dv s}+F_{\dv s}(\dv P\str) \nonumber\\
& = -s_1R_{1,\dv s}-s_2R_{2, \dv s} + [H(X|U_1\str,U_2\str) + s_1R_{1,\dv s} + s_2R_{2, \dv s}] \nonumber\\
&= H(X|U_1\str,U_2\str) \geq D_{\mathrm{CEO}}^{1}(R_{1,\dv s}, R_{2,\dv s}). \label{eq:part1}
\end{align}
Conversely, if $\dv P\str$ is the solution to the minimization in \eqref{eq:RDProblem}, then $I(Y_1;U_1\str|U_2\str)\leq R_{1}$ and $I(Y_2;U_2\str)\leq R_{2}$ and  for any $\dv s$, 
\begin{align*}
D_{\mathrm{CEO}}^{1}(R_1,R_2) =& H(X|U_1\str,U_2\str)\\
\geq& H(X|U_1\str,U_2\str) + s_1(I(Y_1;U_1\str|U_2\str)- R_{1} ) \\
&+ s_2(I(Y_2;U_2\str)- R_{2})\\
=& D_{\dv s} + s_1( R_{1, \dv s}- R_{1} ) + s_2( R_{2, \dv s}- R_{2} ).
\end{align*}
Given $\dv s$, and hence $(R_{1,\dv s},R_{2,\dv s}, D_{\dv s})$, letting $(R_1,R_2) = (R_{1,\dv s}, R_{2,\dv s})$ yields
$D_{\mathrm{CEO}}^{1}(R_{1,\dv s}, R_{2,\dv s}) \geq D_{\dv s}$, which proves, together with \eqref{eq:part1},  statement 1) and 2).
\end{proof}

\vspace{-1.5em}
\subsection{An iterative algorithm to compute $\mathcal{RD}_{\mathrm{CEO}}^{1}$}
\label{sec:algo}
In this section, we derive an algorithm to solve \eqref{eq:Dparam} for a given parameter value $\dv s$. To that end, we express the optimization in \eqref{eq:Dparam} as a minimization of a function $F_{\dv s}(\dv P, \dv Q)$, given in \eqref{eq:CEO_lag2}, over $\dv P$ and some auxiliary pmfs $\dv Q$, defined as $\mathbf{Q} \triangleq \{ Q_{X|U_1,U_2}, Q_{U_1,U_2}, Q_{U_2} \}$. We have the following result.  

\begin{proposition}
\label{lemma:CEO-double}
For each $\dv s \triangleq [s_1, s_2]$, $s_1 > 0$, $s_2 > 0$, the rate-distortion tuple $(D_{\dv s}, R_{1,\dv s}, R_{2,\dv s})$ is given by   
\addtocounter{equation}{1} 
\begin{align}
D_{\dv s} &= -s_1R_{1,\dv s}-s_2R_{2, \dv s} + \argmin_{\dv P} \argmin_{\dv Q} F_{\dv s}(\dv P, \dv Q), \label{eq:CEO-double} 
\end{align}
where $R_{1,\dv s}$ and $R_{2, \dv s}$ are given in \eqref{eq:Rparam} and $\dv P\str$ are the conditional pmfs yielding the minimum in \eqref{eq:Dparam}. 
\end{proposition}
\begin{proof}
Follows from Theorem~\ref{theo:CEO-param} and Lemma~\ref{lemma:CEO-fixed-p} below.
\end{proof}

Motivated by the BA algorithm \cite{Blahut}, we propose an alternate optimization procedure over the set of pmfs $\dv P$ and $\dv Q$ as shown in Algorithm~\ref{algo:CEO}. The steps in the algorithm are derived from the following lemmas. 
   
\begin{lemma}
\label{lemma:CEO-convex}
$F_{\dv s}(\dv P, \dv Q)$ is convex in $\dv P$ and convex in $\dv Q$.
\end{lemma}
\vspace{-3mm}
\begin{proof}
Follows from the log-sum inequality.   
\end{proof}	
	\begin{lemma}
\label{lemma:CEO-fixed-p}
For fixed $\dv P$, there exists a unique $\dv Q$ that achieves the minimum $\argmin_{\dv Q} F_{\dv s}(\dv P, \dv Q) = F_{\dv s}(\dv P)$, given by 
\begin{equation}	
\label{eq:CEO-Q}	
Q_{X|U_1,U_2} = P_{X|U_1,U_2}, \:\: Q_{U_1,U_2} = P_{U_1,U_2}, \:\: Q_{U_2} = P_{U_2}.
\end{equation}	
\end{lemma}
\begin{proof}
The proof follows from the relation
\begin{align}	
&F_{\dv s}(\dv P, \dv Q) - F_{\dv s}(\dv P) \nonumber\\
& =&&\hspace{-7.8em} \sum\nolimits_{u_{1}u_{2}} p(u_{1},u_{2}) D_{\mathrm{KL}}(p(x|u_{1},u_{2})||q(x|u_{1},u_{2})) \nonumber\\ 
& &&\hspace{-7.8em} \!\! + \! s_{1} D_{\mathrm{KL}}(p(u_{1},u_{2})||q(u_{1},u_{2})) \! + \! s_{2} D_{\mathrm{KL}}(p(u_{2})||q(u_{2})) \geq  0, \nonumber
\end{align}	
where equality holds if and only if \eqref{eq:CEO-Q} is satisfied. 
\end{proof}
\vspace{-2mm}
\begin{figure*}[!b]
\hrulefill
\vspace{-0.2em}
\small
\setcounter{MYtempeqncnt}{\value{equation}}
\setcounter{equation}{17}
\begin{align}
F_{\boldsymbol{\beta}}(\dv P,\dv Q) \triangleq&	\: s_{1} \sum_{u_{1}y_{1}} p(u_{1}|y_{1}) p(y_{1}) \log p(u_{1}|y_{1}) + s_{2} \sum_{u_{2}y_{2}} p(u_{2}|y_{2}) p(y_{2}) \log p(u_{2}|y_{2}) + s_{1} \sum_{u_{2}} p(u_{2}) \log p(u_{2}) - s_{2} \sum_{u_{2}} p(u_{2}) \log q(u_{2})  \nonumber\\
& - \alpha \sum_{u_{1}u_{2}y_{1}} p(u_{1},u_{2},y_{1}) \log q(y_{1}|u_{1},u_{2}) - \bar{\alpha} \sum_{u_{1}u_{2}y_{2}} p(u_{1},u_{2},y_{2}) \log q(y_{2}|u_{1},u_{2})  - s_{1} \sum_{u_{1}u_{2}} p(u_{1},u_{2}) \log q(u_{1},u_{2}). \label{eq:MT_lag}
\end{align}
\setcounter{equation}{\value{MYtempeqncnt}}
\vspace{-0.2em}
\hrulefill	
\setcounter{MYtempeqncnt}{\value{equation}}
\setcounter{equation}{20}
\scriptsize
\begin{align}
\mu_{2}(u_{2},y_{2}) \triangleq& \:  \frac{\alpha}{s_{2}} \sum_{u_{1}y_{1}} p(y_{1}|y_{2}) p(u_{1}|y_{1}) \log q(y_{1}|u_{1},u_{2}) + \frac{\bar{\alpha}}{s_{2}} \sum_{u_{1}} p(u_{1}|y_{2}) \log q(y_{2}|u_{1},u_{2}) 
+ \frac{s_{1}}{s_{2}} \sum_{u_{1}} p(u_{1}|y_{2}) \log q(u_{1},u_{2}) + \log q(u_{2}) - \frac{s_{1}}{s_{2}} \!\log p(u_{2}) \nonumber, \\ 
\mu_{1}(u_{1},y_{1}) \triangleq& \: \frac{\alpha}{s_{1}} \sum_{u_{2}} p(u_{2}|y_{1}) \log q(y_{1}|u_{1},u_{2}) + \frac{\bar{\alpha}}{s_{1}} \sum_{u_{2}y_{2}} p(y_{2}|y_{1}) p(u_{2}|y_{2}) \log q(y_{2}|u_{1},u_{2}) + \sum_{u_{2}} p(u_{2}|y_{1}) \log q(u_{1},u_{2}). \label{eq:rho_mt}   
\end{align}
\setcounter{equation}{\value{MYtempeqncnt}}
\end{figure*}

\begin{algorithm}[t!]
\caption{BA-type algorithm to compute $\mathcal{RD}_{\mathrm{CEO}}^{1}$}
\label{algo:CEO}
\begin{algorithmic}[1]
\smallskip
\Inputs{$P_{X,Y_1,Y_2}$, parameters $\dv s$.}
\Outputs{$P_{U_1|Y_1}\str$, $P_{U_2|Y_2}\str$; $(D_{\dv s}, R_{1,\dv s}, R_{2,\dv s})$.}
\Initialize{Set $n=0$. Choose $\dv P^{(0)}$ randomly.\\ Calculate $\dv Q^{(0)}$ by applying steps \ref{algo:ceo-step2} and \ref{algo:ceo-step3}.}
\Repeat 
\State $n \leftarrow n+1$.
\State \label{algo:ceo-step2} Update $\dv P^{(n)}$ by using \eqref{eq:CEO-p}.
\State Update the following pmfs using $\dv P^{(n)}$. 
{\small
\begin{equation*}
\begin{aligned}
p^{(n)}(u_{i}|x) &= \sum \nolimits_{y_{i}} p^{(n)}(u_{i}|y_{i}) p(y_{i}|x), \: &&\hspace{-1.3em} i=1,2,\\
p^{(n)}(u_{i}) &= \sum \nolimits_{y_{i}} p(y_{i}) p^{(n)}(u_{i}|y_{i}), \: &&\hspace{-1.3em} i=1,2,\\
p^{(n)}(u_{1},u_{2},x) &= p(x) p^{(n)}(u_{1}|x) p^{(n)}(u_{2}|x),\\
p^{(n)}(u_{1},u_{2}) &= \sum \nolimits_{x} p^{(n)}(u_{1},u_{2},x).\\  
\end{aligned}
\end{equation*}	
}%
\State \label{algo:ceo-step3} Update $\dv Q^{(n)}$ by using \eqref{eq:CEO-Q}.
\Until{convergence.}
\end{algorithmic}
\end{algorithm}

\begin{lemma}
\label{lemma:CEO-fixed-q}
For fixed $\dv Q$, there exists a unique $\dv P$ that achieves the minimum $\argmin_{\dv P} F_{\dv s}(\dv P, \dv Q)$, where $P_{U_{i}|Y_{i}}$ is given by
\begin{align}
\label{eq:CEO-p}	
p(u_{i}|y_{i}) &= \frac{ \exp [ \rho_{i}(u_{i},y_{i}) ]}{\sum_{u_{i}} \exp [ \rho_{i}(u_{i},y_{i}) ]
},\quad \text{for }\: i=1, 2, 
\end{align}
where $\rho_{i}(u_{i},y_{i})$, $i=1,2$, are defined in \eqref{eq:rho_ceo} given below. 
\end{lemma}
\begin{proof}
We have that $F_{\dv s}(\dv P, \dv Q)$ is convex in $\dv P$ from Lemma~\ref{lemma:CEO-convex}. For a given $\dv Q$ and $\dv s$, in order to minimize $F_{\dv s}(\dv P, \dv Q)$ over the convex set of pmfs $\dv P$, let us define the Lagrangian as      	
\begin{align*}  
\mathcal{L}(\dv P, \boldsymbol\lambda) \triangleq&  F_{\dv s}(\dv P, \dv Q) + \sum \nolimits_{y_{1}} \lambda_{1}(y_{1}) [1 - \sum \nolimits_{u_{1}} p(u_{1}|y_{1}) ] \nonumber\\ 
&+ \sum \nolimits_{y_{2}} \lambda_{2}(y_{2}) [1 - \sum \nolimits_{u_{2}} p(u_{2}|y_{2}) ],
\end{align*}	
where $\lambda_{1}(y_{1}) \geq 0$ and $\lambda_{2}(y_{2}) \geq 0$ are the Lagrange multipliers corresponding the constrains $\sum \nolimits_{u_i} p(u_i|y_i) =1$, $y_i \in \mathcal{Y}_{i}$, $i=1,2$, of the pmfs $P_{U_1|Y_1}$ and $P_{U_2|Y_2}$, respectively. Due to the convexity of $F_{\dv s}(\dv P, \dv Q)$, the KKT conditions are necessary and sufficient for optimality. From the KKT conditions
\addtocounter{equation}{1} 
\begin{equation*}
\frac{\partial \mathcal{L}(\dv P, \boldsymbol\lambda)}{\partial p(u_{1}|y_{1})}  = 0,
\quad \quad
\frac{\partial \mathcal{L}(\dv P, \boldsymbol\lambda)}{\partial p(u_{2}|y_{2})}  = 0, 
\end{equation*}
 we obtain \eqref{eq:CEO-logp} at the bottom of the page.
Then, we proceeded by rearranging \eqref{eq:CEO-logp} as follows 
\addtocounter{equation}{1} 
\begin{align}
\label{eq:CEO-rearrang}
p(u_{i}|y_{i}) &= e^{\tilde{\lambda}_{i}(y_{i})} e^{\rho_{i}(u_{i},y_{i})},\quad i=1,2,   
\end{align}
where $\rho_{i}(u_{i},y_{i})$, $i=1,2$, are given by \eqref{eq:rho_ceo} below, and we define  
$\tilde{\lambda}_{1}(y_{1}) \triangleq \lambda_{1} / [s_{1}p(y_{1})] - 1$ and $
\tilde{\lambda}_{2}(y_{2}) \triangleq [\lambda_{2}(y_{2}) - (s_{1}+s_{2})p(y_{2})] / s_{2}p(y_{2})$.
Note that $\tilde{\lambda}_{i}(y_{i})$ contain all terms independent of $u_{i}$ for $i=1,2$. Finally, the Lagrange multipliers $\lambda_{i}(y_{i})$ satisfying the KKT conditions are obtained by finding $\tilde{\lambda}_{i}(y_{i})$ such that $\sum_{u_{i}} p(u_{i}|y_{i}) = 1$, $i=1,2$. Substituting in \eqref{eq:CEO-rearrang}, $p(u_{i}|y_{i})$ can be found as in \eqref{eq:CEO-p}.   
\end{proof}   

At each iteration of Algorithm~\ref{algo:CEO}, $F_{\dv s}(\dv P^{(n)}, \dv Q^{(n)})$ decreases until eventually it converges. However, since $F_{\dv s}(\dv P, \dv Q)$ is convex in each argument but not necessarily jointly convex, Algorithm~\ref{algo:CEO} does not necessarily converge to the global optimal. In particular, next proposition shows that Algorithm~\ref{algo:CEO} converges to a stationary point of the the minimization in~\eqref{eq:Dparam}.

\begin{proposition}
\label{lemma:CEO-convergence}
The sequence $\{\dv P^{(n)},\dv Q^{(n)}\}$, $n\geq 0$ in Algorithm~\ref{algo:CEO} converges to a stationary solution of the minimization problem in~\eqref{eq:CEO-double} for $n \to\infty$.
\end{proposition}
\begin{proof}
The convergence of the algorithm  follows since due to Lemma~\eqref{lemma:CEO-fixed-p} and Lemma~\eqref{lemma:CEO-fixed-q}, at the $n$-th iteration we have
\begin{align*}
F_{\dv s}( \dv P^{(n-1)}, \dv Q^{(n-1)} )
\geq F_{\dv s}( \dv P^{(n)}, \dv Q^{(n-1)} )
\geq F_{\dv s}( \dv P^{(n)}, \dv Q^{(n)} ), 
\end{align*}
which implies converge since the sequence is lower bounded. The convergence to a stationary point follows by noting that the proposed method is a maximization-minimization algorithm in which $F_{\dv s}( \dv P, \dv Q )$ is a surrogate function \cite{Palomar:2017:MMAlgo}.
\end{proof}


\section{Multiterminal Source Coding Problem}
\label{sec:MT}

In this section, we derive a BA-type algorithm to compute the RD region of the classical two-encoder multiterminal source coding setup, following a similar approach to that in Section~\ref{sec:CEO}. In this setup, we consider two correlated memoryless sources $Y_1$ and $Y_2$ with elements in sets $\mathcal{Y}_1$ and $\mathcal{Y}_2$ and distributed according the joint pmf $P_{Y_1,Y_2}$. The sources $Y_1$ and $Y_2$ are observed at Encoder 1 and 2, each connected to a decoder through an error-free bit-pipe of capacity $R_1$ and $R_2$, respectively. The decoder wants to reproduce an estimate $\hat{Y_{1}}$ and $\hat{Y_{2}}$ of the sources $Y_1$ and $Y_2$ to within some prescribed fidelity levels $D_1$ and $D_2$, respectively; where the distortions are evaluated using the log-loss measure~\eqref{definition-log-loss-distortion-measure}, i.e., $\mathbb{E}[d(Y_1,\hat{Y}_1)] \leq D_1$ and $\mathbb{E}[d(Y_2,\hat{Y}_2)] \leq D_2$. 

The RD region of the two encoder multiterminal source coding problem under log-loss measure is characterized in the following theorem from~\cite[Theorem 6]{Weissman}.
\begin{theorem}{\cite[Theorem 6]{Weissman}}
\label{MT-theorem}
The  tuple  $(R_{1},R_{2},D_{1},D_{2}) \in \mathcal{RD}_{\mathrm{BT}}$ is achievable for the two encoder multiterminal source coding problem under log-loss iff
\begin{align*}
R_{i} & \geq I(U_{i};Y_{i}|U_{i^\mathsf{c}},Q), && \text{for }\: i=1, 2,\\
R_{1} + R_{2} & \geq I(U_{1},U_{2};Y_{1},Y_{2}|Q), \\
D_{i} & \geq H(Y_{i}|U_{1},U_{2},Q), && \text{for }\: i=1, 2.
\end{align*}
for some pmf 
$p(y_{1},y_{2}) p(u_{1}|y_{1},q) p(u_{2}|y_{2},q) p(q)$,   
where $\abs{\mathcal{U}_{1}} \leq \abs{\mathcal{Y}_{1}}$, $\abs{\mathcal{U}_{2}} \leq \abs{\mathcal{Y}_{2}}$, and $\abs{\mathcal{Q}} \leq 5$. 
\end{theorem}

Similarly to Section~\ref{sec:CEO}, first we write $\mathcal{RD}_{\mathrm{BT}}$ in terms of the union of two simpler regions, and then, we propose an algorithm to compute its boundary rate-distortion pairs.  To that end, define the two RD regions $\mathcal{RD}_{\mathrm{BT}}^{i}$, $i=1,2$, as
\vspace{-1mm}
\begin{align*}
\mathcal{RD}_{\mathrm{BT}}^{i} \triangleq \{ &(R_{1},R_{2},D_{2},D_{2}) : \\
&\alpha D_{1} + \bar{\alpha} D_{2} \geq D_{\mathrm{BT},\alpha}^{i}(R_{1},R_{2}), \forall \alpha  \in [0,1] \},\vspace{-1mm}
\end{align*}
where $\bar{\alpha} \triangleq 1 - \alpha$, and
\vspace{-1mm}
\begin{align*}
&D_{\mathrm{BT},\alpha}^{i} (R_{1},R_{2}) \triangleq \argmin\; \alpha H(Y_{1}|U_{1},U_{2}) + \bar{\alpha} H(Y_{2}|U_{1},U_{2}) \\
&\text{s.t.} \quad R_{i} \geq I(Y_{i};U_{i}|U_{i^\mathsf{c}}) \text{ and } R_{i^\mathsf{c}} \geq I(Y_{i^\mathsf{c}};U_{i^\mathsf{c}}),    
\end{align*}
where the optimization is over the set of joint pmfs $P_{U_1,U_2,Y_1,Y_2}$ that satisfy $U_1\mkv Y_1 \mkv Y_2\mkv U_2$.

\begin{proposition}
$\mathcal{RD}_{\mathrm{BT}} = \mathrm{conv} ( \mathcal{RD}_{\mathrm{BT}}^{1} \cup \mathcal{RD}_{\mathrm{BT}}^{2} )$.
\end{proposition}

Now, similarly to Proposition~\ref{lemma:CEO-double}, we provide a parametrization of $\mathcal{RD}_{\mathrm{BT}}^{1}$, which allows to compute each tuple on the boundary of the region as a double minimization over the conditional pmfs $\dv P = \{P_{U_1|Y_1},P_{U_2|Y_2}\}$ and some auxiliary pmfs $\mathbf{Q} \triangleq \{ Q_{Y_{1}|U_1,U_2}, Q_{Y_{2}|U_1,U_2}, Q_{U_1,U_2}, Q_{U_2} \}$,
of an auxiliary function $F_{\boldsymbol{\beta}}(\dv P,\dv Q)$ defined in \eqref{eq:MT_lag}. We have the following result similar to Theorem~\ref{theo:CEO-param},  justified with Lemma~\ref{lemma:MT-fixed-p} below.

\begin{theorem}\label{th:MTParam}
Each tuple on the boundary of $\mc RD_{\mathrm{BT}}^{1}$ can be obtained from some $\boldsymbol{\beta} \triangleq [s_1, s_2, \alpha]$, $s_1 > 0, s_2 > 0, \alpha \in [0,1]$ parametrically as $(R_{1,\boldsymbol{\beta}}, R_{2,\boldsymbol{\beta}}, D_{1,\boldsymbol{\beta}}, D_{2,\boldsymbol{\beta}})$   where
\begin{align}
&\alpha D_{1,\boldsymbol{\beta}} + \bar{\alpha} D_{2,\boldsymbol{\beta}} = -s_1R_{1,\boldsymbol{\beta}}-s_2R_{2, \boldsymbol{\beta}}+ \min_{\dv P} \min_{\dv Q} F_{\boldsymbol{\beta}}(\dv P,\dv Q), \nonumber \\  
&D_{1, \boldsymbol{\beta}} =  H(Y_{1}|U_{1}\str,U_{2}\str), \quad &&\hspace{-13em} D_{2, \boldsymbol{\beta}} =  H(Y_{2}|U_{1}\str,U_{2}\str),  \nonumber \\
&R_{1, \boldsymbol{\beta}} = I(Y_1;U_1\str|U_2\str), \quad &&\hspace{-13em} R_{2, \boldsymbol{\beta}} = I(Y_2;U_2\str),  \nonumber
\end{align}
where $\dv P\str$, $\dv Q\str$ are the pmfs yielding the minimization above.
\end{theorem}

\begin{algorithm}[t!]
\caption{BA-type algorithm to compute $\mathcal{RD}_{\mathrm{BT}}^{1}$}
\label{algo:MT}
\begin{algorithmic}[1]
\smallskip
\Inputs{pmf $P_{Y_1,Y_2}$, parameter $\boldsymbol{\beta}$.}
\Outputs{$P_{U_1|Y_1}\str$, $P_{U_2|Y_2}\str$; $(R_{1,\boldsymbol{\beta}}, R_{2,\boldsymbol{\beta}}, D_{1,\boldsymbol{\beta}}, D_{2,\boldsymbol{\beta}})$.}
\Initialize{Set $n=0$. Choose $\dv P^{(0)}$ randomly.\\ Calculate $\dv Q^{(0)}$ by applying steps \ref{algo:ceo-step2} and \ref{algo:ceo-step3}.}
\Repeat 
\State $n \leftarrow n+1$.
\State \label{algo:mt-step2} Update $\dv P^{(n)}$ by using \eqref{eq:MT-p}.
\State Update the following pmfs. 
{\small
\begin{align*}
p^{(n)}(u_{i}) &= \sum \nolimits_{y_{i}} p(y_{i}) p^{(n)}(u_{i}|y_{i}), &&\hspace{-5em} i=1,2,\\
p^{(n)}(u_{1},u_{2},y_{i}) &= p(y_{i}) p^{(n)}(u_{1}|y_{i}) p^{(n)}(u_{2}|y_{i}), &&\hspace{-5em} i=1,2,\\
p^{(n)}(u_{1},u_{2}) &= \sum \nolimits_{y_{1}y_{2}} p(y_{1},y_{2}) p^{(n)}(u_{1}|y_{1}) p^{(n)}(u_{2}|y_{2}).
\end{align*} 
}%
\State \label{algo:mt-step3}  Update $\dv Q^{(n)}$ by using \eqref{eq:MT-Q}.
\Until{convergence.}
\end{algorithmic}
\end{algorithm}

We have the following lemmas.
\begin{lemma}	
\label{lemma:MT-fixed-p}
For fixed $\dv P$, there exists a unique $\dv Q$ that achieves the minimum $\argmin_{\dv Q} F_{\boldsymbol{\beta}}(\dv P, \dv Q) = F_{\boldsymbol{\beta}}(\dv P, \dv Q\str)$, given by 
\addtocounter{equation}{1} 
\begin{equation}	
\label{eq:MT-Q}	
\begin{aligned}
&Q_{Y_{1}|U_1,U_2} = P_{Y_{1}|U_1,U_2}, \quad Q_{Y_{2}|U_1,U_2} = P_{Y_{2}|U_1,U_2}, \\ 
&Q_{U_1,U_2} = P_{U_1,U_2}, \quad Q_{U_2} = P_{U_2}.
\end{aligned}
\end{equation}	
\end{lemma}

\begin{lemma}
\label{lemma:MT-fixed-q}
For fixed $\dv Q$, there exists a unique $\dv P$ that achieves the minimum $\argmin_{\dv P} F_{ \boldsymbol{\beta}}(\dv P, \dv Q)$, where $P_{U_{i}|Y_{i}}$ is given by
\begin{align}
\label{eq:MT-p}	
p(u_{i}|y_{i}) &= \frac{ \exp [ \mu_{i}(u_{i},y_{i}) ]}{\sum_{u_{i}} \exp [ \mu_{i}(u_{i},y_{i}) ]
},\quad \text{for }\: i=1, 2, 
\end{align}
where $\mu_{i}(u_{i},y_{i})$, $i=1,2$, are defined in \eqref{eq:rho_mt}. 
\end{lemma}

An immediate iterative optimization method follows from the two lemmas above as detailed in Algorithm~\ref{algo:MT}.
Similarly to Algorithm~\ref{algo:CEO}, Algorithm~\ref{algo:MT} converges to a stationary point.
\begin{proposition}
\label{lemma:BT-convergence}
The sequence $\{\dv P^{(n)},\dv Q^{(n)}\}$, $n\geq 0$ in Algorithm~\ref{algo:MT} converges to a stationary point of the minimization problem in Theorem~\ref{th:MTParam} for $n\rightarrow \infty$.
\end{proposition}

\vspace{-6mm}
\section{Numerical Results}
\vspace{-1mm}
\label{sec:results}
In this section, we focus on the computation of the RD region  for a binary CEO setup in which $X$ is 
a Bernoulli random variable distributed as $X \sim \operatorname{Bern}(0.5)$; the channel between the source and  Encoder $i$ is modeled as a binary symmetric channel (BSC) with a crossover probability $\alpha_{i}$ for $i = 1, 2$, i.e., $Y_{i} = X \oplus Z_{i}$, where $Z_{i} \sim \operatorname{Bern}(\alpha_{i})$.

Figure~\ref{fig:3D-sym} shows the rate-distortion tuples of regions $\mc {RD}_{\mathrm{CEO}}^1$ and $\mc {RD}_{\mathrm{CEO}}^2$ computed with Algorithm~\ref{algo:CEO} for a symmetric setup in which $\alpha_1 = \alpha_2 = 0.25$ and different values of $\dv s$. The region $\mc {RD}_{\mathrm{CEO}}$ can be obtained by computing the convex hull of these points. Additionally, the tuples of $\mc{RD}_{\mathrm{CEO}}$ achievable for $R_1=R_2$, i.e., $(R,R,D)$ are shown. 

Figure~\ref{fig:Weissman} shows the rate-distortion tuples computed for $R_{1}=R_{2}=R$ and crossover probabilities $\alpha_{1}=\alpha_{2}=\alpha = \{ 0.01, 0.1, 0.25 \} $. The results coincide with the rate-distortion pairs computed in \cite[Fig. 3]{Weissman} for the same setup by exhaustive search over the conditional pmfs $\dv P$.

Figure~\ref{fig:3D-asym}, illustrates the rate-distortion tuples of the regions $\mc {RD}_{\mathrm{CEO}}^1$ and $\mc {RD}_{\mathrm{CEO}}^2$ for a CEO setup with asymmetric crossover probabilities $\alpha_1 = 0.25$ and $\alpha_2 = 0.1$.

\begin{figure}[t!]
\centering 
\input{fig3D}
\vspace{-1.5em}
\caption{The regions $\mc {RD}_{\mathrm{CEO}}^1$ and $\mc {RD}_{\mathrm{CEO}}^2$ of the CEO setup for crossover probability $\alpha_1= \alpha_2 = 0.25$ and the tuples $(R,R,D)\in \mc{RD}_{\mathrm{CEO}}$.} 
\label{fig:3D-sym}
\vspace{-4.5mm}
\end{figure}

\begin{figure}[t!]
\centering
\input{fig2D}
\vspace{-8mm}
\caption{The RD region of the CEO setup for symmetric rates $R_1 = R_2 = R$ and crossover probability $\alpha = \{0.01, 0.1, 0.25\}$.} 
\label{fig:Weissman}
\vspace{4mm}
\end{figure}
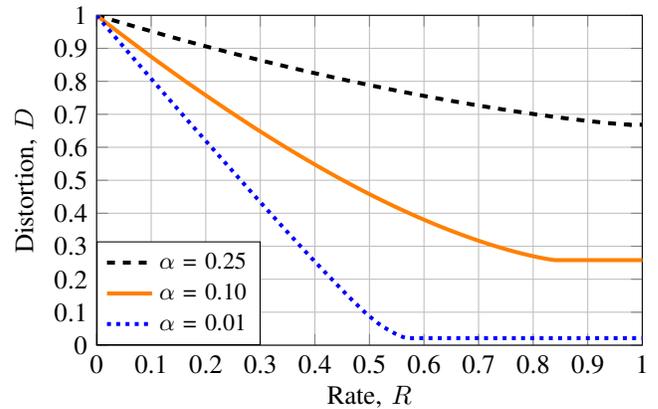

\begin{figure}[t!]
\centering  
\input{fig3D_a}
\vspace{-5mm}
\caption{The regions $\mc {RD}_{\mathrm{CEO}}^1$ and $\mc {RD}_{\mathrm{CEO}}^2$ of the CEO setup for crossover probabilities $\alpha_1 = 0.25$ and $\alpha_2 = 0.1$.} 
\label{fig:3D-asym}
\vspace{-3mm}
\end{figure}
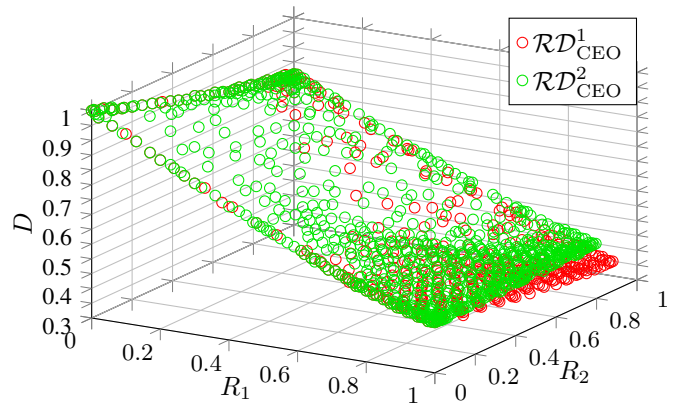

\vspace{-3.5 mm}
\bibliographystyle{ieeetran}
\bibliography{ITWref}

\end{document}

%% file: fig3D.tex
\definecolor{mygreen}{rgb}{0,0.9,0}

\begin{tikzpicture}

\begin{axis}[%
width=0.4\textwidth,
height=0.24\textwidth,
scale only axis,
xmin=0,
xmax=1,
xtick={  0, 0.2, 0.4, 0.6, 0.8,   1},
tick align=outside,
xlabel={$R_{1}$},
xmajorgrids,
ymin=0,
ymax=1,
ytick={  0, 0.2, 0.4, 0.6, 0.8,   1},
ylabel={$R_{2}$},
ymajorgrids,
zmin=0.6,
zmax=1,
ztick={   0, 0.05,  0.1, 0.15,  0.2, 0.25,  0.3, 0.35,  0.4, 0.45,  0.5, 0.55,  0.6, 0.65,  0.7, 0.75,  0.8, 0.85,  0.9, 0.95,    1},
zticklabels={   0, ,  0.1, ,  0.2, ,  0.3, ,  0.4, ,  0.5, ,  0.6, ,  0.7, ,  0.8, ,  0.9, , 1},
zlabel={$D$},
zmajorgrids,
view={52.9}{19.6},
legend style={at={(1,1)},anchor=north east,font=\small},
xlabel shift = -14pt,
ylabel shift = -14pt,
zlabel shift = -4pt,
tick label style={font=\small}, 
]
\addplot3 [color=red,only marks,mark=o,mark options={solid}]
 table[row sep=crcr] {%
1	0.95	0.67\\
0.99	0.95	0.67\\
0.91	0.96	0.67\\
0.54	0.97	0.72\\
0.44	0.98	0.74\\
0.38	0.98	0.75\\
0.03	1	0.81\\
0	1	0.81\\
0.82	0.96	0.68\\
0.49	0.97	0.73\\
0.05	1	0.8\\
0.84	0.96	0.68\\
0.72	0.96	0.7\\
0.64	0.97	0.71\\
0.97	0.95	0.67\\
0.95	0.96	0.67\\
0.89	0.96	0.68\\
0.87	0.96	0.68\\
0.75	0.96	0.69\\
0.59	0.97	0.71\\
0.39	0.98	0.74\\
0.33	0.98	0.75\\
0.79	0.96	0.69\\
0.27	0.98	0.76\\
0.68	0.96	0.7\\
0.55	0.97	0.72\\
0.5	0.97	0.73\\
0.45	0.97	0.73\\
0.14	0.99	0.79\\
0.6	0.97	0.71\\
0.21	0.99	0.77\\
0.9	0.96	0.68\\
0.4	0.98	0.74\\
0.94	0.95	0.67\\
0.51	0.97	0.72\\
0.76	0.96	0.69\\
0.15	0.99	0.78\\
0.01	1	0.81\\
0.83	0.95	0.68\\
0.65	0.96	0.7\\
0.61	0.96	0.71\\
0.41	0.97	0.74\\
1	0.94	0.67\\
0.52	0.97	0.72\\
0.29	0.98	0.76\\
0.09	0.99	0.79\\
0.02	1	0.81\\
0.97	0.94	0.67\\
0.88	0.94	0.68\\
0.7	0.95	0.7\\
0.57	0.96	0.72\\
0.47	0.97	0.73\\
0.36	0.97	0.75\\
0.23	0.98	0.77\\
0.03	0.99	0.81\\
0.99	0.93	0.67\\
0.98	0.93	0.67\\
0.95	0.93	0.67\\
0.94	0.93	0.67\\
0.92	0.93	0.68\\
0.73	0.94	0.7\\
0.62	0.95	0.71\\
0.42	0.96	0.74\\
0.3	0.97	0.76\\
0.1	0.99	0.79\\
0.02	0.99	0.81\\
0	0.99	0.81\\
1	0.92	0.67\\
0.99	0.92	0.67\\
0.98	0.92	0.67\\
0.97	0.92	0.67\\
0.94	0.92	0.67\\
0.88	0.93	0.68\\
0.86	0.93	0.68\\
0.83	0.93	0.68\\
0.62	0.94	0.71\\
0.57	0.95	0.72\\
0.37	0.96	0.75\\
0.17	0.98	0.78\\
1	0.91	0.67\\
0.97	0.91	0.67\\
0.94	0.91	0.68\\
0.77	0.92	0.69\\
0.43	0.95	0.74\\
0.31	0.96	0.76\\
0.25	0.97	0.77\\
0.18	0.97	0.78\\
0.04	0.98	0.81\\
1	0.89	0.67\\
0.99	0.9	0.67\\
0.97	0.9	0.67\\
0.94	0.9	0.68\\
0.92	0.9	0.68\\
0.58	0.93	0.72\\
0.49	0.94	0.73\\
0.43	0.94	0.74\\
0.38	0.95	0.75\\
0.32	0.95	0.76\\
0.18	0.96	0.78\\
0	0.98	0.81\\
1	0.88	0.67\\
0.98	0.88	0.68\\
0.93	0.88	0.68\\
0.91	0.88	0.68\\
0.89	0.89	0.68\\
0.84	0.89	0.69\\
0.78	0.9	0.69\\
0.71	0.9	0.7\\
0.67	0.91	0.71\\
0.63	0.91	0.71\\
0.54	0.92	0.72\\
0.49	0.92	0.73\\
0.38	0.94	0.75\\
0.32	0.94	0.76\\
0.26	0.95	0.77\\
0.13	0.96	0.79\\
0.05	0.97	0.8\\
0	0.97	0.81\\
1	0.86	0.68\\
0.97	0.86	0.68\\
0.96	0.86	0.68\\
0.94	0.86	0.68\\
0.93	0.87	0.68\\
0.91	0.87	0.68\\
0.89	0.87	0.68\\
0.87	0.87	0.69\\
0.81	0.88	0.69\\
0.59	0.9	0.72\\
0.44	0.92	0.74\\
0.39	0.92	0.75\\
0.27	0.94	0.77\\
0.06	0.96	0.8\\
1	0.84	0.68\\
0.94	0.84	0.68\\
0.93	0.85	0.68\\
0.91	0.85	0.68\\
0.89	0.85	0.69\\
0.87	0.85	0.69\\
0.75	0.86	0.7\\
0.72	0.87	0.7\\
0.64	0.88	0.71\\
0.6	0.88	0.72\\
0.55	0.89	0.73\\
0.5	0.9	0.73\\
0.45	0.9	0.74\\
0.34	0.92	0.76\\
0.27	0.92	0.77\\
0.21	0.93	0.78\\
0.13	0.94	0.79\\
0.07	0.95	0.8\\
0.01	0.96	0.81\\
0	0.96	0.81\\
1	0.82	0.68\\
0.98	0.82	0.68\\
0.91	0.83	0.69\\
0.89	0.83	0.69\\
0.87	0.83	0.69\\
0.82	0.84	0.69\\
0.72	0.85	0.71\\
0.69	0.85	0.71\\
0.56	0.87	0.73\\
0.51	0.88	0.73\\
0.46	0.89	0.74\\
0.14	0.93	0.79\\
0.07	0.94	0.8\\
0.01	0.95	0.81\\
0	0.95	0.82\\
1	0.79	0.68\\
0.96	0.8	0.68\\
0.93	0.8	0.69\\
0.89	0.8	0.69\\
0.87	0.81	0.69\\
0.85	0.81	0.69\\
0.82	0.81	0.7\\
0.79	0.82	0.7\\
0.73	0.83	0.71\\
0.61	0.84	0.72\\
0.56	0.85	0.73\\
0.52	0.86	0.73\\
0.46	0.87	0.74\\
0.41	0.87	0.75\\
0.35	0.88	0.76\\
0.29	0.89	0.77\\
0.22	0.9	0.78\\
0.15	0.91	0.79\\
0	0.93	0.82\\
1	0.77	0.69\\
0.99	0.77	0.69\\
0.93	0.77	0.69\\
0.88	0.78	0.69\\
0.85	0.78	0.7\\
0.83	0.79	0.7\\
0.8	0.79	0.7\\
0.77	0.8	0.71\\
0.66	0.81	0.72\\
0.47	0.84	0.74\\
0.42	0.85	0.75\\
0.3	0.87	0.77\\
0.24	0.88	0.78\\
0.17	0.9	0.79\\
0.08	0.91	0.8\\
0.02	0.92	0.82\\
0	0.92	0.82\\
1	0.74	0.69\\
0.95	0.74	0.69\\
0.92	0.75	0.69\\
0.9	0.75	0.7\\
0.85	0.76	0.7\\
0.83	0.76	0.7\\
0.8	0.77	0.71\\
0.77	0.77	0.71\\
0.7	0.78	0.72\\
0.62	0.8	0.73\\
0.53	0.81	0.74\\
0.43	0.83	0.75\\
0.37	0.84	0.76\\
0.31	0.85	0.77\\
0.25	0.86	0.78\\
0.18	0.88	0.79\\
0.1	0.89	0.8\\
0	0.91	0.82\\
1	0.71	0.69\\
0.92	0.72	0.7\\
0.9	0.72	0.7\\
0.83	0.73	0.71\\
0.8	0.74	0.71\\
0.78	0.74	0.71\\
0.74	0.75	0.71\\
0.63	0.77	0.73\\
0.49	0.8	0.75\\
0.44	0.81	0.75\\
0.38	0.82	0.76\\
0.32	0.83	0.77\\
0.26	0.84	0.78\\
0.11	0.87	0.8\\
0	0.89	0.82\\
1	0.68	0.7\\
0.97	0.68	0.7\\
0.95	0.68	0.7\\
0.92	0.69	0.7\\
0.91	0.69	0.7\\
0.86	0.7	0.71\\
0.84	0.7	0.71\\
0.81	0.71	0.71\\
0.78	0.71	0.71\\
0.75	0.72	0.72\\
0.72	0.72	0.72\\
0.68	0.73	0.73\\
0.6	0.75	0.74\\
0.55	0.76	0.74\\
0.45	0.78	0.75\\
0.39	0.79	0.76\\
0.33	0.8	0.77\\
0.27	0.82	0.78\\
0.21	0.83	0.79\\
0.13	0.85	0.8\\
0.06	0.86	0.81\\
0.01	0.87	0.82\\
0	0.87	0.82\\
0	0.87	0.83\\
1	0.64	0.7\\
0.99	0.64	0.7\\
0.98	0.64	0.7\\
0.93	0.65	0.71\\
0.91	0.65	0.71\\
0.87	0.66	0.71\\
0.84	0.67	0.71\\
0.76	0.68	0.72\\
0.72	0.69	0.72\\
0.69	0.7	0.73\\
0.65	0.71	0.73\\
0.61	0.72	0.74\\
0.56	0.73	0.74\\
0.51	0.74	0.75\\
0.46	0.75	0.76\\
0.41	0.76	0.76\\
0.35	0.78	0.77\\
0.29	0.79	0.78\\
0.15	0.82	0.8\\
0.08	0.84	0.81\\
0	0.85	0.83\\
1	0.6	0.71\\
0.93	0.62	0.71\\
0.91	0.62	0.71\\
0.89	0.62	0.71\\
0.87	0.63	0.72\\
0.79	0.64	0.72\\
0.73	0.66	0.73\\
0.69	0.67	0.73\\
0.66	0.67	0.74\\
0.62	0.68	0.74\\
0.57	0.69	0.75\\
0.47	0.72	0.76\\
0.42	0.73	0.77\\
0.36	0.75	0.77\\
0.23	0.78	0.79\\
0.16	0.79	0.8\\
0.08	0.81	0.82\\
0.02	0.83	0.83\\
0	0.83	0.83\\
1	0.57	0.71\\
0.99	0.57	0.71\\
0.96	0.57	0.71\\
0.95	0.58	0.71\\
0.92	0.58	0.72\\
0.9	0.58	0.72\\
0.88	0.59	0.72\\
0.83	0.6	0.72\\
0.77	0.61	0.73\\
0.74	0.62	0.73\\
0.7	0.63	0.74\\
0.67	0.64	0.74\\
0.63	0.65	0.74\\
0.58	0.66	0.75\\
0.54	0.67	0.76\\
0.49	0.69	0.76\\
0.44	0.7	0.77\\
0.32	0.73	0.78\\
0.25	0.75	0.79\\
0.18	0.76	0.8\\
0.12	0.78	0.81\\
0.03	0.8	0.83\\
0	0.81	0.83\\
1	0.53	0.72\\
0.99	0.53	0.72\\
0.97	0.53	0.72\\
0.95	0.54	0.72\\
0.92	0.54	0.72\\
0.88	0.55	0.73\\
0.86	0.55	0.73\\
0.81	0.57	0.73\\
0.75	0.58	0.74\\
0.71	0.59	0.74\\
0.68	0.6	0.74\\
0.64	0.61	0.75\\
0.6	0.62	0.75\\
0.55	0.64	0.76\\
0.5	0.65	0.76\\
0.45	0.66	0.77\\
0.34	0.69	0.79\\
0.27	0.71	0.79\\
0.02	0.78	0.83\\
0	0.78	0.84\\
0	0.79	0.84\\
1	0.48	0.72\\
0.99	0.48	0.72\\
0.99	0.48	0.73\\
0.98	0.49	0.73\\
0.97	0.49	0.73\\
0.95	0.49	0.73\\
0.91	0.5	0.73\\
0.89	0.51	0.73\\
0.81	0.52	0.74\\
0.79	0.53	0.74\\
0.69	0.56	0.75\\
0.65	0.57	0.75\\
0.61	0.58	0.76\\
0.56	0.6	0.76\\
0.51	0.61	0.77\\
0.35	0.66	0.79\\
0.29	0.68	0.8\\
0.23	0.7	0.81\\
0.15	0.72	0.82\\
0.08	0.74	0.83\\
0.01	0.76	0.84\\
0	0.76	0.84\\
1	0.44	0.73\\
0.99	0.44	0.73\\
0.98	0.44	0.73\\
0.97	0.44	0.73\\
0.95	0.45	0.73\\
0.93	0.45	0.74\\
0.87	0.47	0.74\\
0.82	0.48	0.74\\
0.79	0.49	0.75\\
0.76	0.5	0.75\\
0.73	0.51	0.75\\
0.7	0.52	0.75\\
0.62	0.54	0.76\\
0.58	0.55	0.77\\
0.53	0.57	0.77\\
0.48	0.58	0.78\\
0.43	0.6	0.78\\
0.37	0.62	0.79\\
0.31	0.64	0.8\\
0.24	0.66	0.81\\
0.17	0.68	0.82\\
0.09	0.71	0.83\\
0.03	0.73	0.84\\
0	0.73	0.85\\
1	0.39	0.74\\
0.99	0.39	0.74\\
0.98	0.39	0.74\\
0.97	0.4	0.74\\
0.96	0.4	0.74\\
0.93	0.4	0.74\\
0.91	0.41	0.74\\
0.88	0.42	0.75\\
0.8	0.44	0.75\\
0.77	0.45	0.75\\
0.74	0.46	0.76\\
0.67	0.49	0.76\\
0.63	0.5	0.77\\
0.59	0.51	0.77\\
0.54	0.53	0.78\\
0.5	0.54	0.78\\
0.45	0.56	0.79\\
0.39	0.58	0.79\\
0.33	0.6	0.8\\
0.26	0.62	0.81\\
0.19	0.65	0.82\\
0.12	0.67	0.83\\
0	0.7	0.85\\
0	0.71	0.85\\
1	0.34	0.75\\
0.99	0.34	0.75\\
0.98	0.35	0.75\\
0.96	0.35	0.75\\
0.95	0.35	0.75\\
0.93	0.36	0.75\\
0.92	0.36	0.75\\
0.88	0.37	0.75\\
0.86	0.38	0.76\\
0.84	0.38	0.76\\
0.78	0.4	0.76\\
0.75	0.41	0.76\\
0.72	0.42	0.77\\
0.65	0.45	0.77\\
0.61	0.46	0.78\\
0.56	0.48	0.78\\
0.51	0.5	0.79\\
0.46	0.52	0.79\\
0.41	0.54	0.8\\
0.35	0.56	0.81\\
0.22	0.6	0.82\\
0.14	0.63	0.83\\
0.06	0.66	0.85\\
0.01	0.67	0.85\\
0	0.68	0.85\\
1	0.29	0.76\\
0.99	0.29	0.76\\
0.98	0.29	0.76\\
0.97	0.29	0.76\\
0.95	0.3	0.76\\
0.94	0.3	0.76\\
0.91	0.32	0.76\\
0.87	0.32	0.76\\
0.85	0.33	0.77\\
0.82	0.34	0.77\\
0.76	0.36	0.77\\
0.7	0.38	0.78\\
0.66	0.4	0.78\\
0.62	0.41	0.78\\
0.58	0.43	0.79\\
0.53	0.45	0.79\\
0.48	0.47	0.8\\
0.43	0.48	0.8\\
0.38	0.51	0.81\\
0.32	0.53	0.82\\
0.25	0.56	0.83\\
0.18	0.58	0.83\\
0.1	0.61	0.85\\
0	0.65	0.86\\
1	0.23	0.77\\
1	0.24	0.77\\
0.97	0.24	0.77\\
0.95	0.24	0.77\\
0.94	0.25	0.77\\
0.93	0.25	0.77\\
0.91	0.26	0.77\\
0.89	0.26	0.77\\
0.88	0.27	0.77\\
0.85	0.28	0.77\\
0.83	0.29	0.78\\
0.8	0.3	0.78\\
0.78	0.31	0.78\\
0.71	0.33	0.78\\
0.67	0.35	0.79\\
0.64	0.36	0.79\\
0.6	0.38	0.79\\
0.55	0.4	0.8\\
0.5	0.42	0.8\\
0.45	0.44	0.81\\
0.39	0.46	0.81\\
0.34	0.48	0.82\\
0.27	0.51	0.83\\
0.2	0.54	0.84\\
0.12	0.57	0.85\\
0	0.61	0.87\\
0	0	1\\
1	0.18	0.78\\
0.99	0.18	0.78\\
0.98	0.18	0.78\\
0.96	0.19	0.78\\
0.95	0.19	0.78\\
0.93	0.19	0.78\\
0.92	0.2	0.78\\
0.9	0.21	0.78\\
0.88	0.21	0.78\\
0.84	0.23	0.78\\
0.81	0.24	0.79\\
0.76	0.26	0.79\\
0.72	0.28	0.79\\
0.65	0.3	0.8\\
0.61	0.33	0.8\\
0.52	0.37	0.81\\
0.48	0.38	0.81\\
0.42	0.41	0.82\\
0.36	0.43	0.83\\
0.3	0.46	0.83\\
0.23	0.49	0.84\\
0.15	0.52	0.85\\
0.08	0.55	0.86\\
0	0.58	0.87\\
1	0.11	0.79\\
1	0.12	0.79\\
0.98	0.12	0.79\\
0.97	0.12	0.79\\
0.95	0.14	0.79\\
0.92	0.14	0.79\\
0.89	0.16	0.79\\
0.87	0.16	0.79\\
0.85	0.17	0.79\\
0.82	0.18	0.8\\
0.8	0.19	0.8\\
0.77	0.2	0.8\\
0.74	0.22	0.8\\
0.71	0.23	0.8\\
0.67	0.25	0.81\\
0.63	0.26	0.81\\
0.59	0.29	0.81\\
0.55	0.3	0.82\\
0.5	0.33	0.82\\
0.45	0.35	0.83\\
0.39	0.38	0.83\\
0.33	0.4	0.84\\
0.26	0.44	0.85\\
0.18	0.47	0.85\\
0.1	0.5	0.86\\
0	0.55	0.88\\
0	0.54	0.88\\
1	0.04	0.8\\
1	0.06	0.8\\
0.99	0.05	0.8\\
0.97	0.05	0.8\\
0.94	0.08	0.8\\
0.9	0.09	0.8\\
0.86	0.1	0.81\\
0.84	0.11	0.81\\
0.81	0.12	0.81\\
0.78	0.14	0.81\\
0.75	0.15	0.81\\
0.72	0.17	0.81\\
0.69	0.18	0.82\\
0.65	0.2	0.82\\
0.61	0.22	0.82\\
0.56	0.25	0.82\\
0.47	0.3	0.83\\
0.41	0.32	0.84\\
0.35	0.35	0.84\\
0.29	0.38	0.85\\
0.23	0.41	0.86\\
0.13	0.45	0.87\\
-0	0.51	0.89\\
1	0.01	0.81\\
0.99	0.01	0.81\\
0.98	0.01	0.81\\
0.97	0.01	0.81\\
0.96	0.01	0.81\\
0.9	0.03	0.82\\
0.88	0.04	0.82\\
0.86	0.05	0.82\\
0.82	0.07	0.82\\
0.8	0.07	0.82\\
0.77	0.07	0.82\\
0.74	0.09	0.83\\
0.71	0.11	0.83\\
0.67	0.14	0.83\\
0.59	0.18	0.83\\
0.55	0.2	0.84\\
0.49	0.23	0.84\\
0.44	0.25	0.85\\
0.39	0.28	0.85\\
0.32	0.31	0.86\\
0.26	0.35	0.86\\
0.18	0.39	0.87\\
0.08	0.43	0.88\\
0.03	0.46	0.89\\
-0	0.47	0.89\\
1	0	0.81\\
0.99	0	0.81\\
0.98	0	0.81\\
0.95	0	0.82\\
0.93	0	0.82\\
0.92	0	0.82\\
0.91	0	0.82\\
0.89	0	0.82\\
0.87	0	0.82\\
0.85	0.01	0.83\\
0.83	0.01	0.83\\
0.81	0.02	0.83\\
0.78	0.03	0.83\\
0.75	0.04	0.83\\
0.72	0.04	0.84\\
0.61	0.11	0.84\\
0.57	0.12	0.85\\
0.52	0.15	0.85\\
0.47	0.18	0.85\\
0.42	0.21	0.86\\
0.36	0.25	0.86\\
0.29	0.28	0.87\\
0.22	0.32	0.88\\
0.15	0.36	0.88\\
0.04	0.41	0.9\\
0	0.43	0.9\\
0.96	0	0.81\\
0.85	0	0.83\\
0.78	0	0.84\\
0.76	0	0.84\\
0.73	0.01	0.84\\
0.67	0.03	0.85\\
0.63	0.05	0.85\\
0.6	0.04	0.86\\
0.5	0.11	0.86\\
0.45	0.14	0.87\\
0.39	0.18	0.87\\
0.33	0.21	0.88\\
0.26	0.25	0.88\\
0.19	0.29	0.89\\
0.08	0.35	0.9\\
0.02	0.38	0.91\\
-0	0.39	0.91\\
0.83	0	0.83\\
0.73	0	0.85\\
0.68	0	0.85\\
0.61	0.01	0.86\\
0.53	0.05	0.87\\
0.43	0.09	0.88\\
0.37	0.13	0.89\\
0.29	0.18	0.89\\
0.23	0.22	0.9\\
0.14	0.27	0.9\\
-0	0.34	0.92\\
0	0.35	0.92\\
0.97	0	0.81\\
0.87	0	0.83\\
0.81	0	0.83\\
0.58	0	0.87\\
0.41	0.05	0.9\\
0.34	0.1	0.9\\
0.26	0.15	0.9\\
0.19	0.19	0.91\\
0.11	0.23	0.92\\
-0	0.3	0.93\\
0.37	0.03	0.91\\
0.31	0.07	0.91\\
0.23	0.11	0.92\\
0.16	0.15	0.93\\
0.08	0.2	0.93\\
-0	0.25	0.94\\
-0	0.26	0.94\\
0.04	0.17	0.95\\
0	0.21	0.95\\
-0	0.16	0.96\\
-0	0.15	0.96\\
0.08	0.03	0.97\\
-0	0.11	0.97\\
0	0.06	0.99\\
-0	0.05	0.99\\
0	0.01	1\\
0.47	0	0.89\\
0.06	0	0.99\\
0.39	0	0.91\\
0.3	0	0.93\\
0.16	0	0.96\\
0.25	-0	0.94\\
0.01	-0	1\\
0.2	-0	0.95\\
0.61	0	0.87\\
0.1	-0	0.97\\
0.09	-0	0.98\\
0.54	-0	0.88\\
0.79	0	0.84\\
0.03	-0	0.99\\
0.65	0	0.86\\
0.11	0	0.97\\
0.43	-0	0.9\\
0.15	-0	0.96\\
0.51	-0	0.89\\
0.71	-0	0.85\\
0.34	-0	0.92\\
0.35	-0	0.92\\
0.05	-0	0.99\\
0.56	0.97	0.72\\
0.35	0.98	0.75\\
0.46	0.97	0.73\\
0.28	0.98	0.76\\
0.53	0.95	0.72\\
0.24	0.97	0.77\\
0.91	0.9	0.68\\
0.2	0.94	0.78\\
0.98	0.84	0.68\\
0.4	0.89	0.75\\
0.79	0.84	0.7\\
0.59	0.78	0.73\\
0.2	0.73	0.81\\
0.86	0.51	0.73\\
0.9	0.36	0.75\\
0.69	0.43	0.77\\
0.03	0.64	0.86\\
0.74	0.32	0.78\\
0.26	-0	0.94\\
0.21	-0	0.95\\
-0	0.1	0.97\\
0.04	-0	0.99\\
};
\addlegendentry{$\mathcal{RD}_{\mathrm{CEO}}^{1}$};

\addplot3 [color=mygreen,only marks,mark=o,mark options={solid}]
 table[row sep=crcr] {%
0.95	1	0.67\\
0.77	1	0.69\\
0.71	1	0.69\\
0.64	1	0.7\\
0.6	1	0.71\\
0.48	1	0.72\\
0.44	1	0.73\\
0	1	0.81\\
0.94	1	0.67\\
0.91	1	0.67\\
0.88	1	0.67\\
0.84	1	0.68\\
0.82	1	0.68\\
0.68	1	0.7\\
0.34	1	0.75\\
0.23	1	0.77\\
0.18	1	0.78\\
0.01	1	0.81\\
0.93	1	0.67\\
0.86	1	0.68\\
0.57	1	0.71\\
0.53	1	0.72\\
0.39	1	0.74\\
0.12	1	0.79\\
0.06	1	0.8\\
0.74	1	0.69\\
0.29	1	0.76\\
0.79	1	0.68\\
0.92	1	0.67\\
0.95	0.99	0.67\\
0.9	0.99	0.67\\
0.93	0.99	0.67\\
0.92	0.99	0.67\\
0.88	0.99	0.67\\
0.84	0.99	0.68\\
0.71	0.99	0.69\\
0.68	0.99	0.7\\
0.24	1	0.77\\
0.95	0.98	0.67\\
0.94	0.98	0.67\\
0.93	0.98	0.67\\
0.92	0.98	0.67\\
0.91	0.98	0.67\\
0.82	0.99	0.68\\
0.64	0.99	0.7\\
0.9	0.98	0.67\\
0.8	0.98	0.68\\
0.68	0.98	0.7\\
0.64	0.98	0.7\\
0.44	0.99	0.73\\
0.39	0.99	0.74\\
0.06	0.99	0.8\\
0.01	0.99	0.81\\
0	0.99	0.81\\
0.96	0.96	0.67\\
0.93	0.97	0.67\\
0.91	0.97	0.67\\
0.9	0.97	0.67\\
0.88	0.97	0.68\\
0.53	0.98	0.72\\
0.39	0.98	0.74\\
0.24	0.99	0.77\\
0.96	0.95	0.67\\
0.93	0.95	0.67\\
0.91	0.95	0.67\\
0.88	0.96	0.68\\
0.86	0.96	0.68\\
0.68	0.97	0.7\\
0.61	0.97	0.71\\
0.57	0.97	0.71\\
0.49	0.97	0.73\\
0.44	0.97	0.73\\
0.34	0.98	0.75\\
0.29	0.98	0.76\\
0.24	0.98	0.77\\
0.18	0.98	0.78\\
0.12	0.98	0.79\\
0.96	0.93	0.67\\
0.92	0.94	0.67\\
0.88	0.94	0.68\\
0.86	0.94	0.68\\
0.8	0.95	0.69\\
0.71	0.95	0.7\\
0.65	0.96	0.7\\
0.61	0.96	0.71\\
0.49	0.96	0.73\\
0.4	0.97	0.74\\
0.35	0.97	0.75\\
0.24	0.97	0.77\\
0.01	0.98	0.81\\
0	0.98	0.81\\
0.96	0.91	0.67\\
0.96	0.92	0.67\\
0.93	0.92	0.68\\
0.91	0.92	0.68\\
0.9	0.92	0.68\\
0.87	0.93	0.68\\
0.85	0.93	0.68\\
0.82	0.93	0.68\\
0.77	0.93	0.69\\
0.72	0.94	0.7\\
0.53	0.95	0.72\\
0.49	0.95	0.73\\
0.45	0.95	0.73\\
0.4	0.96	0.74\\
0.35	0.96	0.75\\
0.18	0.97	0.78\\
0.07	0.97	0.8\\
0.01	0.97	0.81\\
0	0.97	0.81\\
0.96	0.89	0.68\\
1	0	0.81\\
0.96	0.9	0.68\\
0.95	0.9	0.68\\
0.88	0.91	0.68\\
0.85	0.91	0.68\\
0.83	0.91	0.69\\
0.78	0.92	0.69\\
0.58	0.93	0.72\\
0.45	0.94	0.73\\
0.4	0.94	0.74\\
0.3	0.95	0.76\\
0.24	0.95	0.77\\
0.19	0.96	0.78\\
0.13	0.96	0.79\\
0.07	0.96	0.8\\
0.96	0.87	0.68\\
0.95	0.88	0.68\\
0.9	0.88	0.68\\
0.87	0.89	0.68\\
0.85	0.89	0.69\\
0.83	0.89	0.69\\
0.8	0.89	0.69\\
0.75	0.9	0.7\\
0.72	0.9	0.7\\
0.69	0.91	0.7\\
0.66	0.91	0.71\\
0.36	0.93	0.75\\
0.25	0.94	0.77\\
0.07	0.95	0.8\\
0.02	0.96	0.81\\
0	0.96	0.81\\
0.96	0.85	0.68\\
0.92	0.86	0.68\\
0.9	0.86	0.68\\
0.87	0.86	0.69\\
0.85	0.87	0.69\\
0.83	0.87	0.69\\
0.81	0.87	0.69\\
0.78	0.87	0.69\\
0.75	0.88	0.7\\
0.72	0.88	0.7\\
0.69	0.89	0.7\\
0.62	0.89	0.71\\
0.5	0.91	0.73\\
0.46	0.91	0.74\\
0.31	0.92	0.76\\
0.25	0.93	0.77\\
0.19	0.93	0.78\\
0.02	0.95	0.81\\
0	0.95	0.82\\
0.96	0.82	0.68\\
0.95	0.83	0.68\\
0.92	0.83	0.69\\
0.85	0.84	0.69\\
0.83	0.84	0.69\\
0.81	0.85	0.69\\
0.78	0.85	0.7\\
0.76	0.86	0.7\\
0.73	0.86	0.7\\
0.7	0.86	0.71\\
0.66	0.87	0.71\\
0.42	0.9	0.74\\
0.31	0.91	0.76\\
0.26	0.91	0.77\\
0.2	0.92	0.78\\
0.14	0.92	0.79\\
0.02	0.93	0.81\\
0	0.93	0.82\\
0.96	0.78	0.69\\
0.95	0.8	0.69\\
0.93	0.8	0.69\\
0.89	0.81	0.69\\
0.84	0.82	0.7\\
0.81	0.82	0.7\\
0.79	0.82	0.7\\
0.76	0.83	0.7\\
0.73	0.83	0.71\\
0.7	0.84	0.71\\
0.63	0.85	0.72\\
0.59	0.85	0.72\\
0.56	0.86	0.73\\
0.42	0.88	0.75\\
0.32	0.89	0.76\\
0.26	0.89	0.77\\
0.09	0.91	0.8\\
0	0.92	0.82\\
0.96	0.75	0.69\\
0.96	0.76	0.69\\
0.91	0.78	0.69\\
0.88	0.78	0.69\\
0.86	0.78	0.7\\
0.84	0.79	0.7\\
0.82	0.79	0.7\\
0.79	0.8	0.7\\
0.77	0.8	0.71\\
0.74	0.81	0.71\\
0.7	0.81	0.71\\
0.67	0.82	0.72\\
0.64	0.82	0.72\\
0.6	0.83	0.72\\
0.56	0.83	0.73\\
0.38	0.86	0.76\\
0.27	0.88	0.77\\
0.22	0.88	0.78\\
0.16	0.89	0.79\\
0	0.91	0.82\\
0.96	0.72	0.7\\
0.91	0.74	0.7\\
0.82	0.76	0.7\\
0.8	0.76	0.71\\
0.77	0.77	0.71\\
0.74	0.77	0.71\\
0.71	0.78	0.71\\
0.68	0.79	0.72\\
0.64	0.79	0.72\\
0.57	0.81	0.73\\
0.48	0.82	0.74\\
0.44	0.83	0.75\\
0.38	0.84	0.76\\
0.33	0.85	0.77\\
0.09	0.88	0.81\\
0.03	0.89	0.82\\
0	0.89	0.82\\
0.96	0.68	0.7\\
0.92	0.71	0.7\\
0.89	0.71	0.7\\
0.8	0.73	0.71\\
0.78	0.74	0.71\\
0.75	0.74	0.71\\
0.72	0.75	0.72\\
0.69	0.76	0.72\\
0.61	0.77	0.73\\
0.57	0.78	0.73\\
0.53	0.79	0.74\\
0.49	0.79	0.75\\
0.44	0.8	0.75\\
0.39	0.81	0.76\\
0.34	0.82	0.77\\
0.29	0.83	0.78\\
0.23	0.84	0.78\\
0.17	0.85	0.79\\
0.05	0.86	0.82\\
0	0.87	0.82\\
0	0.87	0.83\\
0.97	0.64	0.71\\
0.97	0.65	0.71\\
0.96	0.65	0.7\\
0.95	0.66	0.7\\
0.87	0.68	0.71\\
0.83	0.69	0.71\\
0.78	0.7	0.72\\
0.75	0.71	0.72\\
0.72	0.72	0.72\\
0.69	0.72	0.72\\
0.62	0.74	0.73\\
0.58	0.75	0.74\\
0.54	0.76	0.74\\
0.45	0.77	0.75\\
0.29	0.8	0.78\\
0.24	0.81	0.79\\
0.01	0.85	0.83\\
0	0.85	0.83\\
0.97	0.59	0.71\\
0.96	0.61	0.71\\
0.92	0.63	0.71\\
0.9	0.63	0.71\\
0.84	0.65	0.72\\
0.79	0.66	0.72\\
0.73	0.68	0.73\\
0.7	0.69	0.73\\
0.67	0.69	0.73\\
0.63	0.7	0.74\\
0.59	0.71	0.74\\
0.55	0.72	0.75\\
0.51	0.73	0.75\\
0.36	0.76	0.77\\
0.31	0.78	0.78\\
0.25	0.79	0.79\\
0.19	0.8	0.8\\
0.13	0.81	0.81\\
0.01	0.83	0.83\\
0	0.83	0.83\\
0.97	0.54	0.72\\
0.88	0.6	0.72\\
0.82	0.62	0.72\\
0.8	0.62	0.73\\
0.77	0.63	0.73\\
0.71	0.65	0.73\\
0.68	0.66	0.74\\
0.64	0.66	0.74\\
0.6	0.68	0.74\\
0.56	0.69	0.75\\
0.47	0.71	0.76\\
0.32	0.74	0.78\\
0.26	0.76	0.79\\
0.2	0.77	0.8\\
0.07	0.8	0.82\\
0.02	0.81	0.83\\
0	0.81	0.83\\
0.97	0.49	0.73\\
0.97	0.5	0.73\\
0.97	0.51	0.73\\
0.97	0.52	0.72\\
0.96	0.52	0.72\\
0.93	0.54	0.72\\
0.89	0.55	0.73\\
0.87	0.56	0.73\\
0.8	0.58	0.73\\
0.75	0.59	0.74\\
0.72	0.6	0.74\\
0.68	0.62	0.74\\
0.65	0.62	0.74\\
0.61	0.63	0.75\\
0.57	0.65	0.75\\
0.53	0.66	0.76\\
0.49	0.67	0.76\\
0.44	0.68	0.77\\
0.39	0.7	0.78\\
0.33	0.71	0.78\\
0.27	0.73	0.79\\
0.21	0.74	0.8\\
0.15	0.75	0.81\\
0.02	0.78	0.83\\
0	0.78	0.84\\
0	0.79	0.84\\
0.98	0.44	0.74\\
0.97	0.46	0.73\\
0.95	0.48	0.73\\
0.92	0.49	0.73\\
0.91	0.5	0.73\\
0.9	0.5	0.73\\
0.88	0.51	0.73\\
0.84	0.52	0.74\\
0.76	0.55	0.74\\
0.73	0.56	0.74\\
0.66	0.58	0.75\\
0.62	0.59	0.75\\
0.58	0.61	0.76\\
0.54	0.62	0.76\\
0.5	0.63	0.77\\
0.45	0.65	0.77\\
0.4	0.66	0.78\\
0.35	0.67	0.79\\
0.29	0.69	0.79\\
0.23	0.71	0.8\\
0.17	0.72	0.81\\
0.11	0.74	0.82\\
0.04	0.75	0.83\\
0	0.76	0.84\\
0.98	0.38	0.74\\
0.98	0.39	0.74\\
0.98	0.4	0.74\\
0.97	0.41	0.74\\
0.96	0.42	0.74\\
0.95	0.43	0.74\\
0.93	0.44	0.74\\
0.9	0.45	0.74\\
0.89	0.46	0.74\\
0.87	0.46	0.74\\
0.82	0.48	0.74\\
0.77	0.5	0.75\\
0.74	0.51	0.75\\
0.67	0.54	0.75\\
0.64	0.55	0.76\\
0.59	0.57	0.76\\
0.55	0.58	0.77\\
0.51	0.59	0.77\\
0.47	0.6	0.78\\
0.36	0.64	0.79\\
0.25	0.67	0.81\\
0.19	0.69	0.81\\
0.12	0.7	0.82\\
0.05	0.72	0.84\\
0.01	0.73	0.84\\
0	0.73	0.85\\
0.98	0.32	0.75\\
0.98	0.34	0.75\\
0.96	0.37	0.75\\
0.95	0.38	0.75\\
0.91	0.39	0.75\\
0.89	0.4	0.75\\
0.87	0.41	0.75\\
0.83	0.43	0.75\\
0.78	0.45	0.75\\
0.75	0.46	0.76\\
0.72	0.47	0.76\\
0.69	0.49	0.76\\
0.65	0.5	0.76\\
0.61	0.51	0.77\\
0.57	0.53	0.77\\
0.53	0.54	0.78\\
0.48	0.56	0.78\\
0.43	0.58	0.79\\
0.38	0.59	0.79\\
0.32	0.61	0.8\\
0.27	0.63	0.81\\
0.2	0.65	0.82\\
0.06	0.69	0.84\\
0.01	0.7	0.85\\
0	0.71	0.85\\
0.99	0.25	0.77\\
0.98	0.26	0.77\\
0.98	0.26	0.76\\
0.98	0.27	0.76\\
0.98	0.28	0.76\\
0.98	0.29	0.76\\
0.98	0.3	0.76\\
0.97	0.3	0.76\\
0.94	0.32	0.76\\
0.92	0.34	0.76\\
0.88	0.36	0.76\\
0.86	0.36	0.76\\
0.82	0.38	0.76\\
0.79	0.39	0.76\\
0.73	0.42	0.77\\
0.66	0.45	0.77\\
0.63	0.46	0.77\\
0.58	0.48	0.78\\
0.54	0.5	0.78\\
0.49	0.52	0.79\\
0.45	0.53	0.79\\
0.4	0.55	0.8\\
0.29	0.59	0.81\\
0.22	0.61	0.82\\
0.16	0.63	0.83\\
0.02	0.67	0.85\\
0	0.68	0.85\\
0.99	0.19	0.78\\
0.97	0.25	0.77\\
0.95	0.26	0.77\\
0.91	0.28	0.77\\
0.87	0.3	0.77\\
0.83	0.32	0.77\\
0.8	0.34	0.77\\
0.78	0.35	0.77\\
0.75	0.36	0.77\\
0.71	0.38	0.78\\
0.68	0.39	0.78\\
0.64	0.41	0.78\\
0.6	0.43	0.78\\
0.56	0.45	0.79\\
0.52	0.46	0.79\\
0.47	0.48	0.8\\
0.42	0.5	0.8\\
0.36	0.52	0.81\\
0.31	0.54	0.82\\
0.24	0.57	0.82\\
0.18	0.59	0.83\\
0.11	0.61	0.84\\
0.04	0.63	0.85\\
0	0.64	0.86\\
0	0.65	0.86\\
1	0.03	0.81\\
0.99	0.13	0.79\\
0.99	0.15	0.78\\
0.97	0.18	0.78\\
0.96	0.18	0.78\\
0.96	0.19	0.78\\
0.93	0.21	0.78\\
0.9	0.22	0.78\\
0.88	0.24	0.78\\
0.86	0.25	0.78\\
0.84	0.26	0.78\\
0.82	0.27	0.78\\
0.79	0.29	0.78\\
0.7	0.34	0.79\\
0.66	0.35	0.79\\
0.62	0.37	0.79\\
0.54	0.41	0.8\\
0.49	0.43	0.8\\
0.44	0.45	0.81\\
0.38	0.47	0.81\\
0.33	0.49	0.82\\
0.26	0.52	0.83\\
0.2	0.55	0.84\\
0.14	0.57	0.85\\
0.04	0.6	0.86\\
0.01	0.61	0.86\\
0	0.61	0.87\\
1	0.04	0.8\\
1	0.05	0.8\\
1	0.07	0.8\\
0.99	0.09	0.79\\
0.98	0.1	0.79\\
0.95	0.13	0.79\\
0.94	0.14	0.79\\
0.91	0.15	0.79\\
0.9	0.17	0.79\\
0.85	0.2	0.79\\
0.8	0.22	0.79\\
0.75	0.25	0.79\\
0.71	0.27	0.8\\
0.68	0.29	0.8\\
0.64	0.31	0.8\\
0.6	0.33	0.8\\
0.55	0.35	0.81\\
0.51	0.37	0.81\\
0.46	0.39	0.81\\
0.4	0.42	0.82\\
0.35	0.44	0.83\\
0.29	0.47	0.83\\
0.23	0.49	0.84\\
0.15	0.52	0.85\\
0.07	0.56	0.86\\
-0	0.58	0.87\\
1	0.02	0.81\\
1	0.01	0.81\\
0.99	0.03	0.81\\
0.98	0.03	0.81\\
0.97	0.04	0.81\\
0.96	0.06	0.8\\
0.94	0.07	0.8\\
0.91	0.08	0.8\\
0.87	0.13	0.8\\
0.82	0.15	0.8\\
0.79	0.16	0.8\\
0.76	0.19	0.8\\
0.7	0.22	0.81\\
0.66	0.25	0.81\\
0.62	0.27	0.81\\
0.58	0.28	0.81\\
0.53	0.31	0.82\\
0.48	0.34	0.82\\
0.43	0.36	0.83\\
0.38	0.38	0.83\\
0.32	0.42	0.84\\
0.26	0.44	0.84\\
0.19	0.47	0.85\\
0.12	0.5	0.86\\
0.04	0.53	0.87\\
0	0.54	0.88\\
0	0	1\\
0.99	0	0.81\\
0.98	0	0.81\\
0.97	0	0.81\\
0.96	0.01	0.81\\
0.95	0.01	0.81\\
0.93	0.02	0.81\\
0.92	0.01	0.82\\
0.9	0.03	0.82\\
0.88	0.04	0.82\\
0.86	0.04	0.82\\
0.84	0.05	0.82\\
0.81	0.1	0.81\\
0.78	0.1	0.82\\
0.75	0.12	0.82\\
0.72	0.16	0.82\\
0.68	0.17	0.82\\
0.65	0.19	0.82\\
0.6	0.22	0.82\\
0.56	0.24	0.83\\
0.51	0.27	0.83\\
0.46	0.3	0.83\\
0.41	0.32	0.84\\
0.34	0.36	0.84\\
0.29	0.38	0.85\\
0.21	0.42	0.86\\
0.14	0.45	0.87\\
0.01	0.5	0.88\\
0	0.51	0.89\\
0.96	0	0.81\\
0.95	0	0.82\\
0.93	0	0.82\\
0.92	0	0.82\\
0.89	0	0.82\\
0.87	0.01	0.82\\
0.85	0.01	0.83\\
0.83	0.02	0.83\\
0.77	0.04	0.83\\
0.74	0.06	0.83\\
0.7	0.1	0.83\\
0.67	0.11	0.83\\
0.63	0.14	0.83\\
0.59	0.17	0.84\\
0.49	0.23	0.84\\
0.44	0.26	0.85\\
0.38	0.29	0.85\\
0.31	0.32	0.86\\
0.25	0.36	0.86\\
0.18	0.39	0.87\\
0.11	0.42	0.88\\
0.04	0.45	0.89\\
0	0.47	0.89\\
0.91	0	0.82\\
0.87	0	0.82\\
0.85	0	0.83\\
0.83	0	0.83\\
0.81	0	0.83\\
0.76	0.01	0.84\\
0.73	0.03	0.84\\
0.69	0.05	0.84\\
0.66	0.06	0.85\\
0.62	0.08	0.85\\
0.52	0.15	0.85\\
0.46	0.19	0.85\\
0.41	0.22	0.86\\
0.35	0.25	0.86\\
0.29	0.29	0.87\\
0.21	0.33	0.88\\
0.15	0.36	0.88\\
-0	0.43	0.9\\
0.76	0	0.84\\
0.73	0	0.84\\
0.7	0	0.85\\
0.67	0.01	0.85\\
0.64	0.02	0.86\\
0.6	0.04	0.86\\
0.55	0.07	0.86\\
0.51	0.09	0.87\\
0.44	0.15	0.87\\
0.38	0.19	0.87\\
0.32	0.21	0.88\\
0.25	0.26	0.88\\
0.18	0.29	0.89\\
0.11	0.33	0.9\\
-0	0.39	0.91\\
0.79	0	0.84\\
0.73	0	0.85\\
0.71	0	0.85\\
0.68	0	0.85\\
0.64	0	0.86\\
0.61	0	0.87\\
0.58	0.01	0.87\\
0.53	0.02	0.88\\
0.48	0.07	0.88\\
0.43	0.08	0.88\\
0.37	0.13	0.89\\
0.29	0.19	0.89\\
0.22	0.23	0.9\\
0.14	0.27	0.9\\
0.06	0.31	0.91\\
0	0.35	0.92\\
-0	0.34	0.92\\
0.58	0	0.87\\
0.5	0.01	0.89\\
0.35	0.08	0.9\\
0.26	0.15	0.9\\
0.19	0.19	0.91\\
0.09	0.25	0.92\\
-0	0.3	0.93\\
0.65	0	0.86\\
0.31	0.07	0.91\\
0.24	0.11	0.92\\
0.06	0.22	0.93\\
-0	0.25	0.94\\
0.87	0	0.83\\
0.21	0.07	0.93\\
-0	0.2	0.95\\
-0	0.21	0.95\\
0.18	0.04	0.95\\
0.1	0.08	0.96\\
-0	0.15	0.96\\
0	0.16	0.96\\
0	0.11	0.97\\
-0	0.09	0.98\\
0.02	0.03	0.99\\
-0	0.06	0.99\\
0	0.01	1\\
0.01	0	1\\
0.54	0	0.88\\
0.11	0	0.97\\
0.34	-0	0.92\\
0.21	0	0.95\\
0.51	-0	0.89\\
0.43	0	0.9\\
0.16	-0	0.96\\
0.3	0	0.93\\
0.06	-0	0.99\\
0.25	-0	0.94\\
0.03	-0	0.99\\
0.26	-0	0.94\\
0.47	0	0.89\\
0.05	-0	0.99\\
0.39	-0	0.91\\
0.35	0	0.92\\
0.94	0.99	0.67\\
0.18	0.99	0.78\\
0.74	0.98	0.69\\
0.87	0.91	0.68\\
0.19	0.95	0.78\\
0.96	0.79	0.69\\
0.93	0.77	0.69\\
0.28	0.85	0.77\\
0.35	0.79	0.77\\
0.11	0.84	0.81\\
0.88	0.64	0.71\\
0.42	0.72	0.77\\
0.97	0.45	0.73\\
0.3	0.65	0.8\\
0.95	0.32	0.76\\
0.99	0.23	0.77\\
0.58	0.39	0.79\\
0.03	0.37	0.91\\
0	0.26	0.94\\
0.1	-0	0.97\\
0.15	-0	0.96\\
};
 \addlegendentry{$\mathcal{RD}_{\mathrm{CEO}}^{2}$};

\addplot3 [color=black,only marks,mark=x,mark options={solid},mark size={1pt}]
 table[row sep=crcr] {%
0.95	1	0.67\\
0.95	0.99	0.67\\
0.94	0.99	0.67\\
0.94	0.98	0.67\\
0.93	0.98	0.67\\
0.92	0.97	0.67\\
0.91	0.96	0.67\\
0.9	0.94	0.68\\
0.89	0.93	0.68\\
0.87	0.91	0.68\\
0.86	0.9	0.68\\
0.84	0.88	0.69\\
0.82	0.86	0.69\\
0.8	0.84	0.7\\
0.78	0.81	0.7\\
0.76	0.79	0.71\\
0.73	0.77	0.71\\
0.71	0.74	0.72\\
0.68	0.71	0.73\\
0.66	0.69	0.73\\
0.63	0.66	0.74\\
0.61	0.63	0.75\\
0.58	0.6	0.76\\
0.55	0.57	0.77\\
0.52	0.54	0.78\\
0.49	0.51	0.79\\
0.46	0.48	0.8\\
0.43	0.45	0.81\\
0.4	0.41	0.82\\
0.37	0.38	0.83\\
0.34	0.35	0.85\\
0.31	0.31	0.86\\
0.27	0.28	0.87\\
0.24	0.24	0.89\\
0.2	0.21	0.9\\
0.17	0.17	0.92\\
0.13	0.13	0.94\\
0.09	0.1	0.95\\
0.08	0.04	0.97\\
0.04	0	0.99\\
0.01	0	1\\
0	0	1\\
0.6	0.63	0.75\\
0.43	0.44	0.81\\
0.34	0.34	0.85\\
0.3	0.32	0.86\\
0.13	0.14	0.94\\
0.1	0.09	0.95\\
0.07	0.04	0.97\\
0.01	0.03	0.99\\
0	0.01	1\\
0.75	0.79	0.71\\
0.73	0.76	0.71\\
0.66	0.68	0.73\\
0.3	0.31	0.86\\
0.02	0.14	0.96\\
0	0.08	0.98\\
0.03	0.01	0.99\\
0.85	0.9	0.68\\
0.46	0.47	0.8\\
0.27	0.27	0.87\\
0.2	0.2	0.9\\
0.14	0.13	0.94\\
0.8	0.83	0.7\\
0.77	0.81	0.7\\
0.66	0.68	0.74\\
0.63	0.65	0.74\\
0.57	0.6	0.76\\
0.37	0.38	0.84\\
0.33	0.34	0.85\\
0.27	0.28	0.88\\
0.16	0.17	0.92\\
0.12	0.14	0.94\\
0.06	0.04	0.98\\
0.04	0.01	0.99\\
0.85	0.89	0.68\\
0.82	0.85	0.69\\
0.7	0.74	0.72\\
0.65	0.68	0.74\\
0.54	0.57	0.77\\
0.49	0.5	0.79\\
0.18	0.02	0.95\\
0.01	0.07	0.98\\
0.84	0.87	0.69\\
0.7	0.73	0.72\\
0.57	0.59	0.76\\
0.52	0.53	0.78\\
0.39	0.41	0.82\\
0.37	0.37	0.84\\
0.27	0.27	0.88\\
0.23	0.24	0.89\\
0.2	0.21	0.91\\
0.04	0.04	0.98\\
0.91	0.95	0.67\\
0.88	0.93	0.68\\
0.85	0.89	0.69\\
0.83	0.87	0.69\\
0.79	0.83	0.7\\
0.75	0.78	0.71\\
0.6	0.62	0.75\\
0.54	0.56	0.77\\
0.42	0.44	0.81\\
0.2	0.2	0.91\\
0.13	0.03	0.96\\
0.06	0.05	0.97\\
0.02	0.01	0.99\\
0.81	0.85	0.69\\
0.73	0.76	0.72\\
0.62	0.65	0.74\\
0.51	0.53	0.78\\
0.36	0.37	0.84\\
0.1	0.08	0.95\\
0	0.1	0.98\\
0.48	0.5	0.79\\
0.45	0.47	0.8\\
0.43	0.43	0.81\\
0.39	0.4	0.82\\
0.19	0.21	0.91\\
0.17	0.16	0.92\\
0.13	0.02	0.96\\
0	0.03	0.99\\
0.93	0.97	0.67\\
0.72	0.76	0.72\\
0.67	0.71	0.73\\
0.4	0.4	0.83\\
0.3	0.3	0.86\\
0.26	0.27	0.88\\
0.12	0.13	0.94\\
0.1	0.08	0.96\\
0.05	0.04	0.98\\
0.68	0.7	0.73\\
0.42	0.43	0.81\\
0.39	0.4	0.83\\
0.24	0.23	0.89\\
0.2	0.19	0.91\\
0.14	0.12	0.94\\
0.08	0.03	0.97\\
0.03	0.03	0.99\\
0.92	0.96	0.67\\
0.62	0.65	0.75\\
0.23	0.23	0.89\\
0.2	0	0.95\\
0.09	0.09	0.96\\
0.05	0.06	0.98\\
0.88	0.92	0.68\\
0.67	0.7	0.73\\
0.65	0.67	0.74\\
0.59	0.62	0.75\\
0.29	0.3	0.86\\
0.08	0.11	0.96\\
0.05	0.06	0.97\\
0.02	0.03	0.99\\
0.77	0.8	0.7\\
0.72	0.75	0.72\\
0.45	0.46	0.8\\
0.35	0.37	0.84\\
0.09	0.1	0.96\\
0.06	0.03	0.98\\
0.03	0	0.99\\
0.64	0.67	0.74\\
0.59	0.61	0.75\\
0.56	0.59	0.76\\
0.33	0.33	0.85\\
0.29	0.31	0.86\\
0.13	0.12	0.94\\
0.01	0.13	0.96\\
0.74	0.78	0.71\\
0.57	0.58	0.76\\
0.48	0.49	0.79\\
0.42	0.43	0.82\\
0.32	0.34	0.85\\
0.19	0.2	0.91\\
0.17	0.15	0.92\\
0	0.04	0.99\\
0.86	0.91	0.68\\
0.7	0.72	0.72\\
0.62	0.64	0.75\\
0.59	0.62	0.76\\
0.56	0.58	0.76\\
0.53	0.56	0.77\\
0.51	0.52	0.78\\
0.36	0.36	0.84\\
0.32	0.33	0.85\\
0.29	0.3	0.87\\
0.26	0.26	0.88\\
0.03	0.12	0.96\\
0.05	0.03	0.98\\
0.79	0.82	0.7\\
0.76	0.8	0.71\\
0.69	0.73	0.72\\
0.59	0.61	0.76\\
0.54	0.55	0.77\\
0.5	0.52	0.78\\
0.22	0.24	0.89\\
0.16	0.16	0.92\\
0.04	0.05	0.98\\
0.02	0	0.99\\
0.74	0.77	0.71\\
0.69	0.72	0.72\\
0.53	0.55	0.77\\
0.41	0.43	0.82\\
0.12	0.04	0.96\\
0.81	0.84	0.7\\
0.44	0.46	0.81\\
0.42	0.42	0.82\\
0.38	0.4	0.83\\
0.35	0.36	0.84\\
0	0.19	0.95\\
0.02	0.02	0.99\\
0.78	0.82	0.7\\
0.56	0.58	0.77\\
0.07	0.11	0.96\\
0.02	0.08	0.97\\
0.01	0.01	0.99\\
0.71	0.75	0.72\\
0.67	0.69	0.73\\
0.61	0.64	0.75\\
0.47	0.49	0.8\\
0.38	0.39	0.83\\
0.22	0.23	0.89\\
0.19	0.19	0.91\\
0.69	0.72	0.73\\
0.58	0.61	0.76\\
0.53	0.55	0.78\\
0.5	0.52	0.79\\
0.29	0.29	0.87\\
0.18	0.2	0.91\\
0.07	0.02	0.98\\
0.64	0.66	0.74\\
0.08	0.1	0.96\\
0.03	0.04	0.98\\
0.01	0.01	1\\
0.44	0.45	0.81\\
0.23	0.02	0.94\\
0.03	0.18	0.95\\
0.05	0.05	0.98\\
0.5	0.51	0.79\\
0.47	0.48	0.8\\
0.41	0.42	0.82\\
0.26	0.25	0.88\\
0.23	0.22	0.9\\
0.15	0.16	0.92\\
0.1	0.07	0.96\\
0.07	0.01	0.98\\
0.55	0.58	0.77\\
0.53	0.54	0.78\\
0.37	0.39	0.83\\
0.32	0.33	0.86\\
0.22	0.23	0.9\\
0.01	0.04	0.99\\
0.52	0.55	0.78\\
0.49	0.52	0.79\\
0.35	0.35	0.84\\
0.31	0.33	0.86\\
0.28	0.29	0.87\\
0.25	0.27	0.88\\
0.16	0.16	0.93\\
0.03	0.07	0.98\\
0.34	0.36	0.84\\
0.32	0.32	0.86\\
0.25	0.26	0.88\\
0.15	0.16	0.93\\
0.01	0.19	0.95\\
0.09	0.08	0.96\\
0.16	0.15	0.93\\
0.08	0.09	0.96\\
0.08	0.01	0.98\\
0.02	0.04	0.99\\
0.31	0.32	0.86\\
0.21	0.23	0.9\\
0.05	0.02	0.98\\
0.4	0.42	0.82\\
0.38	0.38	0.83\\
0.25	0.25	0.88\\
0.22	0.22	0.9\\
0.18	0.19	0.91\\
0.11	0.05	0.96\\
0.44	0.44	0.81\\
0.11	0.13	0.94\\
0.13	0.11	0.94\\
0.05	0.11	0.96\\
0.03	0.02	0.99\\
0.34	0.35	0.84\\
0.32	0.31	0.86\\
0.11	0.14	0.94\\
0.1	0.06	0.96\\
0.04	0.06	0.98\\
0.28	0.28	0.87\\
0.24	0.26	0.88\\
0.14	0.16	0.93\\
0.33	0.35	0.85\\
0.27	0.29	0.87\\
0.08	0.08	0.96\\
0.22	0.21	0.9\\
0.13	0.01	0.97\\
0.24	0.25	0.88\\
0.21	0.22	0.9\\
0.19	0.18	0.91\\
0.12	0.11	0.94\\
0.12	0.03	0.96\\
0.24	0.25	0.89\\
0.16	0.14	0.93\\
0.09	0.07	0.96\\
0.25	0.25	0.89\\
0.33	0	0.92\\
0.02	0.13	0.96\\
0.07	0.09	0.96\\
0.08	0	0.98\\
0	0.02	0.99\\
};
 \addlegendentry{$R_{1} = R_{2}$};

\end{axis}
\end{tikzpicture}

%% file: fig2D.tex
\begin{tikzpicture} 

\begin{axis}[
width=0.4\textwidth,
height=0.242\textwidth,
scale only axis,
xmin=0,
xmax=1,
xlabel={Rate, $R$},
xmajorgrids,
x label style={above=-1.2em},
xtick={0, 0.1, 0.2, 0.3, 0.4, 0.5, 0.6, 0.7, 0.8, 0.9, 1},
xticklabels={0, 0.1, 0.2, 0.3, 0.4, 0.5, 0.6, 0.7, 0.8, 0.9, 1},
ymin=0,
ymax=1,
ylabel={Distortion, $D$},
ymajorgrids,
ytick={0, 0.1, 0.2, 0.3, 0.4, 0.5, 0.6, 0.7, 0.8, 0.9, 1},
yticklabels={0, 0.1, 0.2, 0.3, 0.4, 0.5, 0.6, 0.7, 0.8, 0.9, 1},
y label style={above=-1.6em},
legend style={at={(0,0)},anchor=south west, font=\small}],
tick label style={font=\small}, 
]   

\addplot [color=black,dashed,line width=1.5pt]
table[row sep=crcr]{%
	2.51796509811542e-10	0.99999999601823\\
	2.76342661104238e-10	0.999999992265016\\
	8.34891515976311e-09	0.999999751087988\\
	1.59161286027558e-08	0.999999532899218\\
	2.12988089485419e-08	0.999999469502423\\
	2.80321186738505e-08	0.999999405767417\\
	4.68520568748217e-08	0.999999283023249\\
	8.59391327430301e-08	0.999999081628179\\
	1.41671629982054e-07	0.999998796188587\\
	7.01267131840902e-07	0.99999748089558\\
	9.90778185915541e-05	0.999939422275709\\
	0.00102679625582983	0.999456988539225\\
	0.00570974051191576	0.997121157931461\\
	0.0267398808839861	0.986639743591317\\
	0.0536184501309293	0.97374133214428\\
	0.124310399432001	0.940241594201981\\
	0.133350800270454	0.936063308633115\\
	0.150582599262627	0.928189259791466\\
	0.161880502964701	0.923078095980117\\
	0.165990921995289	0.921228160553482\\
	0.197992987191314	0.906992573764455\\
	0.204532325923197	0.904122383880231\\
	0.21680300869367	0.898777071867027\\
	0.223903365839672	0.895706553405086\\
	0.226197974984936	0.894715499113374\\
	0.227936327135779	0.893966081612804\\
	0.229679653058511	0.893216784880572\\
	0.230267217356632	0.892964337006935\\
	0.23145048690803	0.892456198713494\\
	0.233239241402087	0.891689246428042\\
	0.235057749774116	0.890909905573605\\
	0.238131546380272	0.889595536364465\\
	0.240633265162775	0.888528207384162\\
	0.252316017150146	0.883567221152725\\
	0.254767866131168	0.882531476660586\\
	0.257244414879764	0.881488426991306\\
	0.259137347700701	0.880691359203916\\
	0.260413944130024	0.88015426315242\\
	0.26234272356128	0.879344126861711\\
	0.264286009490857	0.87852975173279\\
	0.266919554737988	0.877427280798755\\
	0.268919981736145	0.876591444198502\\
	0.271624502167358	0.875463099308698\\
	0.272985593737041	0.874896726950115\\
	0.280968362941245	0.871582657592538\\
	0.282273763573211	0.871042991092708\\
	0.284920819459649	0.869949382820156\\
	0.286920257760455	0.869125901719501\\
	0.290322118132746	0.86772524981904\\
	0.294468897573266	0.866024369838389\\
	0.295869346944325	0.865451323129888\\
	0.298694617242317	0.864297660567394\\
	0.299412478941554	0.86400459953831\\
	0.300847671666232	0.86341921645166\\
	0.304478622889111	0.86194209853871\\
	0.308184028817952	0.86043722359065\\
	0.313050041787099	0.858467935951749\\
	0.316597186090873	0.857037892230814\\
	0.318761420541129	0.856166206050108\\
	0.323136553271196	0.854410657691972\\
	0.325357028053107	0.85352186800657\\
	0.32759792155216	0.852626457717948\\
	0.329860331051729	0.851724092249662\\
	0.332901597497124	0.850514793860626\\
	0.333670817769577	0.850209008491971\\
	0.335213547946713	0.84959660656276\\
	0.337549343701864	0.848670408025221\\
	0.339117942780341	0.848049442457346\\
	0.345702701190172	0.845450254894519\\
	0.348727871653729	0.844262301973352\\
	0.351026127490099	0.843360790905392\\
	0.353345513024434	0.84245256551222\\
	0.354124590088466	0.842147611768852\\
	0.354906912180715	0.841841406999067\\
	0.355686058131596	0.84153753033757\\
	0.357253767402942	0.840926334853087\\
	0.359618238607255	0.840006429025174\\
	0.363621447383185	0.838450247035246\\
	0.36849261990425	0.836564721066205\\
	0.372604776249729	0.834980387591901\\
	0.373438091140175	0.83465937350809\\
	0.377430947246065	0.833124380606126\\
	0.38063398530728	0.831896792427963\\
	0.385490318332685	0.830043909301114\\
	0.387950878980066	0.82910718772208\\
	0.392089240047523	0.8275369770678\\
	0.396279618420434	0.825952599304941\\
	0.397970781750716	0.825314677309013\\
	0.401377694680469	0.824032378818736\\
	0.405679770447004	0.822419039665644\\
	0.40726989705982	0.821823649974757\\
	0.409761248492854	0.820892894764423\\
	0.412270007543396	0.819957744137659\\
	0.413950917158096	0.819332596421569\\
	0.416493203222243	0.818387403722695\\
	0.418189470524874	0.817759771036618\\
	0.420748647828444	0.816814378978681\\
	0.425943807772574	0.814895974631613\\
	0.431195200710962	0.812968126431496\\
	0.433844452996869	0.811999362872529\\
	0.43845317126701	0.810316493361861\\
	0.439310347686767	0.81000514336731\\
	0.443632151570801	0.808436745969594\\
	0.44624376727829	0.807492405913516\\
	0.44799825944616	0.806858040753261\\
	0.449756891525722	0.806223838908899\\
	0.455956566999447	0.803998338644879\\
	0.461345248094877	0.80207119349806\\
	0.462250127748595	0.801748162543981\\
	0.46496560098856	0.800782836982853\\
	0.465873663757569	0.800460530970607\\
	0.467697116418874	0.799813594270798\\
	0.469346052772209	0.799229301487098\\
	0.471119835018002	0.798601015138198\\
	0.475573027918615	0.797030368193976\\
	0.476470095272382	0.796714190021931\\
	0.477366497670477	0.796399121602003\\
	0.48006228163241	0.795453874353348\\
	0.48186573568607	0.794823038142868\\
	0.482770962332373	0.794506407030023\\
	0.48367764054955	0.794189572215871\\
	0.484586487981812	0.793871983025871\\
	0.48822710205091	0.792605039863185\\
	0.491897175839424	0.7913299603581\\
	0.49557800039699	0.790058647361512\\
	0.501156275914519	0.788134445969814\\
	0.507209463729748	0.786059855866806\\
	0.509951146451752	0.785125440987473\\
	0.513625844813045	0.783875905228451\\
	0.516391138228643	0.782939442120427\\
	0.519164026234024	0.782003770366743\\
	0.523813289676395	0.780438135640209\\
	0.524743599230474	0.780126554862898\\
	0.531302444823323	0.777931359343421\\
	0.533865269288265	0.777074069017199\\
	0.535709987535458	0.776459893608338\\
	0.537557263006483	0.775846376470654\\
	0.53940472080487	0.775235363971909\\
	0.542185723707404	0.77431645238962\\
	0.548699114433481	0.772176175548426\\
	0.552439731062533	0.770952446824375\\
	0.557135098898392	0.769421377362503\\
	0.560899144617438	0.768200780720506\\
	0.567466918755965	0.766082989838142\\
	0.570255028566511	0.765188362906771\\
	0.571184880104392	0.764890905880003\\
	0.575842314336896	0.763404899241005\\
	0.581453365259212	0.761621214019033\\
	0.582389353902865	0.761325034491485\\
	0.584261977870391	0.760733197183563\\
	0.588013159853966	0.759550565634758\\
	0.590415652953457	0.75879613099813\\
	0.593192641936208	0.757926515520121\\
	0.595974023461036	0.757057504799431\\
	0.59782740572866	0.756480987863218\\
	0.599682724433772	0.755904622977483\\
	0.60432484853214	0.754467468592618\\
	0.607113986319716	0.753606144039192\\
	0.608042600760163	0.753320593450898\\
	0.613616032217048	0.751613183842135\\
	0.620286884479758	0.749579817240818\\
	0.621207201322609	0.749300602459703\\
	0.624879882350129	0.748194038686357\\
	0.631320649563592	0.746255172145018\\
	0.634994914507574	0.745157060050187\\
	0.638662095701959	0.744068182411225\\
	0.641413430383409	0.743251295531672\\
	0.649588641816625	0.740842968438069\\
	0.653212901935757	0.7397858329604\\
	0.655930641379415	0.738993750105112\\
	0.657741186135524	0.738466916437916\\
	0.661350096017487	0.737424451165914\\
	0.666750961403886	0.735870625309347\\
	0.670340462006957	0.734842839276553\\
	0.680015691273754	0.732095704623397\\
	0.683552647398694	0.731105226840411\\
	0.687964347364624	0.729870102794041\\
	0.696726359088756	0.727435841076751\\
	0.705289053727241	0.725089118366885\\
	0.707874306679971	0.724386915785991\\
	0.713873517755898	0.722764024334527\\
	0.718975656346503	0.721391239393567\\
	0.723191802362162	0.720266729475651\\
	0.729071007641995	0.718710511427916\\
	0.734895632722226	0.717184017679315\\
	0.737374480827494	0.716536085566599\\
	0.73819772495809	0.716321719055756\\
	0.742294544794742	0.715257567461072\\
	0.743108664557339	0.7150478735832\\
	0.744731910992756	0.714631299841991\\
	0.747155351988598	0.714011465852515\\
	0.75541685145789	0.711899522442323\\
	0.75701292794935	0.711496196011373\\
	0.761762267569376	0.710303815550656\\
	0.763334034558021	0.709909623159011\\
	0.765680216081894	0.709322336187152\\
	0.768010227934167	0.708742760287857\\
	0.778478122366078	0.706166383634326\\
	0.779240159820343	0.705979358666025\\
	0.782266943531867	0.705243709810384\\
	0.788233938660289	0.703805679606624\\
	0.791172799684198	0.703104039244231\\
	0.796245684025863	0.701899296539575\\
	0.802846408150992	0.700346359912881\\
	0.808494966371153	0.699034613456998\\
	0.812649393393734	0.698077298700009\\
	0.814017891898548	0.697764108124527\\
	0.817400395543356	0.696998881364297\\
	0.818070770212239	0.696847982501229\\
	0.827307410252536	0.69477591050115\\
	0.833776990178198	0.693333651711956\\
	0.83567465254496	0.69291793193552\\
	0.838795256067912	0.692237856217753\\
	0.846592052277544	0.690562928347439\\
	0.850800458859151	0.689659217265797\\
	0.851981525092554	0.689410898752729\\
	0.85547275720914	0.688679983325059\\
	0.85775835124103	0.688201528290381\\
	0.86761538369621	0.686171732879363\\
	0.872959716828323	0.685098021227516\\
	0.88499119571015	0.682727198670739\\
	0.886449915451391	0.682444642551868\\
	0.89299793548377	0.681193339089261\\
	0.901474194693518	0.679613170347849\\
	0.903596657359992	0.679219832388818\\
	0.908079036784123	0.678409859393005\\
	0.917512784409119	0.676747814669668\\
	0.918949904393779	0.676499161671981\\
	0.920007069040738	0.676316900715514\\
	0.921389591267239	0.676078569838319\\
	0.933303100722996	0.674099283964117\\
	0.934159367314509	0.673957635099858\\
	0.934721130993347	0.673868911784544\\
	0.94279356974486	0.672603017734625\\
	0.944450309361158	0.672350357372673\\
	0.945587059024948	0.672179944547701\\
	0.951919842927674	0.671241328087058\\
	0.952301400911913	0.671185148666095\\
	0.961026068178179	0.669979028100173\\
	0.966124609392223	0.669323733205427\\
	0.966224005490176	0.669311542926514\\
	0.966517011742099	0.669275665409592\\
	0.970725709159358	0.668774948611687\\
	0.97127887903393	0.66871283029621\\
	0.973808495202332	0.668438798241811\\
	0.973954499676581	0.668423015307982\\
	0.975609639677335	0.668258886236235\\
	0.975741540462166	0.668246109177276\\
	0.976638856604184	0.668166568274939\\
	0.976654632703041	0.668165197434042\\
	0.977043181510969	0.668134122042608\\
	0.977184823685164	0.668124193396539\\
	0.977213604615719	0.668122424615329\\
	0.977216890039576	0.668122251159804\\
	0.977217000911086	0.668122246013295\\
	0.977217001462459	0.668122245993301\\
	0.977217001462482	0.668122245993301\\
	0.977237748062929	0.668259198205593\\
	0.97724529599408	0.668330530543051\\
	0.97722209871153	0.668378683088767\\
	1	0.66838\\
};
\addlegendentry{$\alpha\text{ = 0.25}$};

\addplot [color=orange,solid,line width=1.5pt]
table[row sep=crcr]{%
	1.78221425549605e-08	0.999999829305674\\
	2.03651664409964e-08	0.999999768938361\\
	1.05171797364423e-07	0.999998746098115\\
	1.76604584321287e-07	0.999997986701116\\
	2.47079432702116e-07	0.99999738412722\\
	4.11350569542575e-07	0.999996182785806\\
	1.23905490099953e-06	0.999992398382571\\
	4.2271994569222e-06	0.999986419603814\\
	7.03741825693e-06	0.999981201296463\\
	5.78702037734244e-05	0.999904730771263\\
	0.000321492961927137	0.999549626008545\\
	0.00112991242432673	0.998474753042659\\
	0.00651115538636646	0.991501792688633\\
	0.0120137269601927	0.984423874814965\\
	0.0344128349379211	0.955926308264965\\
	0.087225151707249	0.890302249328445\\
	0.095561654376193	0.880017820146474\\
	0.107213184249745	0.865835052093494\\
	0.119868832423952	0.850704142884921\\
	0.165691316567289	0.796484114407059\\
	0.21096502823366	0.745189431353631\\
	0.213199524891158	0.742672374308737\\
	0.217771883651123	0.737524702838426\\
	0.221889390193298	0.732894456122158\\
	0.228669692836407	0.725323549462379\\
	0.233340279210701	0.720131709940256\\
	0.238702616743787	0.714177704036264\\
	0.245677626832238	0.706471537616921\\
	0.255701750937213	0.695475027880567\\
	0.258777773806097	0.692131149170906\\
	0.261097425808628	0.689611262541903\\
	0.288975341358908	0.659374478459243\\
	0.314495232195943	0.632843757693336\\
	0.324917785714418	0.622076153023022\\
	0.334926446030867	0.611831189150214\\
	0.344736963460843	0.601886155089308\\
	0.347112868585394	0.599499157486591\\
	0.349636722534608	0.59696812094618\\
	0.352090159909781	0.594511184042239\\
	0.356622839155294	0.589973283890849\\
	0.36334582360897	0.583283886619502\\
	0.369770620077865	0.576929249128632\\
	0.372280393256778	0.574476204595962\\
	0.384797723453929	0.562242772382521\\
	0.395920085021091	0.551573488383084\\
	0.397701725867137	0.549867658960427\\
	0.400367805968077	0.547318724376261\\
	0.403220984094893	0.544605027569353\\
	0.411093552241861	0.53716999142033\\
	0.418799109884775	0.529957534808677\\
	0.427531116905265	0.5218469995442\\
	0.437110023188711	0.513039007991451\\
	0.442004157684654	0.50858719888892\\
	0.453395278744833	0.498349243024134\\
	0.460755496647368	0.491824358047988\\
	0.463086304574277	0.489764744931758\\
	0.465724087654151	0.487439234950751\\
	0.470215153857428	0.48349635664037\\
	0.473385113674483	0.480735229679438\\
	0.475666467543161	0.478749381955472\\
	0.479703905326993	0.475253998878811\\
	0.480380715114208	0.474669974665288\\
	0.488616738011564	0.467571995428518\\
	0.496311230232748	0.461074755643424\\
	0.500004245036758	0.457970901131539\\
	0.502409681249532	0.455953820283339\\
	0.505411413133746	0.453453081659855\\
	0.508759476547696	0.450667675752547\\
	0.512940551113241	0.447215063184772\\
	0.518284335024606	0.442831308350801\\
	0.523932579569929	0.438242380948382\\
	0.527880182144346	0.435051699915028\\
	0.530028495088255	0.433329272292129\\
	0.532321455276464	0.431493340710759\\
	0.536940489394394	0.427821666222339\\
	0.542850798106959	0.423135753707422\\
	0.547827342491632	0.419259613975737\\
	0.549817278441297	0.417712869296003\\
	0.551823942163389	0.41615967890184\\
	0.55645665723357	0.412588383848583\\
	0.558122879325041	0.411306737666947\\
	0.563673308529502	0.407070138562747\\
	0.564978582078622	0.406075308938473\\
	0.569377168801846	0.402767050429822\\
	0.57152788927507	0.401158940456515\\
	0.577088398521329	0.397014236260216\\
	0.581352141531315	0.393859339952499\\
	0.585852984863545	0.390579603434584\\
	0.590103784614676	0.387495991316167\\
	0.592258304521713	0.385941708200656\\
	0.593522535811181	0.385029926150575\\
	0.594538624975246	0.38430360461709\\
	0.596263580535751	0.383070993085214\\
	0.597842278633125	0.38194537004574\\
	0.599673972018668	0.380642582235095\\
	0.602600561056628	0.37856723294737\\
	0.604099836092643	0.377505904746432\\
	0.605072837054319	0.376821243211583\\
	0.608887283746864	0.374155668413036\\
	0.612108472857684	0.371917558153418\\
	0.614388952131763	0.37034506321945\\
	0.616950240378697	0.368585945733882\\
	0.620137018056717	0.366403307996475\\
	0.625354284679034	0.362862410196706\\
	0.629810576822088	0.359886726185276\\
	0.631366339140872	0.358849227793171\\
	0.634900679101158	0.356509154903805\\
	0.636001636807798	0.3557837617358\\
	0.639939694731532	0.353202704388553\\
	0.641686266431342	0.352064539334034\\
	0.642557218003891	0.351501442739002\\
	0.647122685783971	0.34855210204774\\
	0.649242448878368	0.347199929637737\\
	0.652727525066886	0.34497886119654\\
	0.653414197204199	0.344545920280972\\
	0.654106816265101	0.344110739301231\\
	0.655836257832825	0.343024696790396\\
	0.657876561577487	0.341743705336851\\
	0.659746552032019	0.340574929067728\\
	0.66147396556753	0.339503165849196\\
	0.664199378042472	0.337818943758884\\
	0.666048906944885	0.336679899452641\\
	0.671398030196332	0.333417061478725\\
	0.675081545718551	0.331206271622666\\
	0.675997953041952	0.330656386444691\\
	0.676877394487176	0.330130192505141\\
	0.679546088545401	0.328546349149221\\
	0.681333174050661	0.327487887222524\\
	0.684914175786904	0.325390403185003\\
	0.685829014391219	0.324855636880853\\
	0.687532262831772	0.3238678135988\\
	0.688444390751125	0.323340122087238\\
	0.691010715240808	0.321859556946425\\
	0.692610943592222	0.320946689990058\\
	0.694438330459526	0.319908023890236\\
	0.697818508611769	0.317990970120461\\
	0.700846001010576	0.316291202164478\\
	0.702561435436587	0.315341364056935\\
	0.703931016777367	0.314585149482263\\
	0.70631531286988	0.313270342664269\\
	0.707274345572916	0.312742409854517\\
	0.709238034399215	0.311673047834\\
	0.712119634314094	0.310109166570474\\
	0.713707213665845	0.309256236614016\\
	0.716992299565128	0.307503732658779\\
	0.718502266728996	0.306703139231809\\
	0.721896325048059	0.304918560620424\\
	0.723440946195391	0.304113256303586\\
	0.725443623135188	0.30307461271357\\
	0.726923625470941	0.302314166676662\\
	0.728076281423047	0.301722151748513\\
	0.729841053398879	0.300822013335268\\
	0.730868792616194	0.300299085705629\\
	0.73190851069859	0.299772588951133\\
	0.733634930758154	0.298904645712899\\
	0.737654521639049	0.296897331708238\\
	0.73909776335436	0.296185101761986\\
	0.741132241618295	0.295190639673827\\
	0.741342176170601	0.295088383119424\\
	0.742778194490865	0.294389945003984\\
	0.744721906086416	0.293446731391882\\
	0.745849908210228	0.292905280610606\\
	0.747220804098201	0.292249729152252\\
	0.748157135727601	0.29180390787171\\
	0.75107066141022	0.290422340149085\\
	0.752551755474633	0.289724983348778\\
	0.757383247658434	0.287488518216854\\
	0.759830377281623	0.286375468770289\\
	0.761121570447499	0.285793056531974\\
	0.76243303250462	0.285203867820254\\
	0.763657475439137	0.284654900663806\\
	0.767183230855875	0.283094180111302\\
	0.769619644721261	0.282030658681424\\
	0.77201134682037	0.280995313148416\\
	0.773009385678905	0.280567369615725\\
	0.775693199893235	0.279427581803805\\
	0.777315459827359	0.278748081522012\\
	0.77944241615829	0.277867882195927\\
	0.781616090739942	0.276970522971886\\
	0.784024105890383	0.275993724389115\\
	0.785809146285537	0.275276131177881\\
	0.787748628603764	0.274501837965273\\
	0.789505097308902	0.273813019174864\\
	0.790990464880614	0.273234618729072\\
	0.794172315863318	0.272011873577899\\
	0.795526893517622	0.271499956975337\\
	0.797662119935214	0.270697807169887\\
	0.798992786178183	0.27020428178202\\
	0.800683534727954	0.269587291830756\\
	0.8011643287968	0.269413051315303\\
	0.802987690346579	0.268754044142832\\
	0.804690639753149	0.268148335465218\\
	0.805212888195947	0.267965021515725\\
	0.807896332360141	0.267027605112543\\
	0.809178111968399	0.266588107865669\\
	0.810439245697565	0.266160308635681\\
	0.81240680558337	0.265501930854879\\
	0.814348160628762	0.264861847770858\\
	0.815211224261085	0.264584637866328\\
	0.817662850502715	0.263800462222364\\
	0.81794930338584	0.263709716456254\\
	0.81934271317326	0.263276424938\\
	0.819859200324634	0.263116940873046\\
	0.821465585041076	0.262624976600062\\
	0.823790711901031	0.261934494519769\\
	0.82449976203932	0.26172747146294\\
	0.824769308265579	0.261649578717725\\
	0.826058126755754	0.261278867037628\\
	0.826455811416636	0.261166869149855\\
	0.827685408065527	0.260826782818437\\
	0.82890333045998	0.260491671049886\\
	0.829543436906681	0.260318331106196\\
	0.830734797439056	0.260003325250299\\
	0.831003714552269	0.259933090783896\\
	0.83196833344073	0.259684418901021\\
	0.832335819909087	0.259590988158198\\
	0.83315046183062	0.259386973795556\\
	0.834663567905568	0.259020992619809\\
	0.834858113739475	0.258974706118631\\
	0.835847988840757	0.258747608951259\\
	0.836543171588144	0.258591032881679\\
	0.836846119547262	0.258524232891663\\
	0.837125290204013	0.258463951932058\\
	0.837369690573898	0.258411969796954\\
	0.8376124027873	0.258361168610049\\
	0.837764432183424	0.258329366887316\\
	0.838022420519348	0.258276667089287\\
	0.838186987921807	0.258243194226372\\
	0.838628148938329	0.258156773224712\\
	0.838858623146639	0.258112722950312\\
	0.839051955272207	0.258076725097225\\
	0.839175959517266	0.25805403793744\\
	0.839357992683574	0.258021740550687\\
	0.839419646071213	0.258010961851614\\
	0.839550912055845	0.257988889835394\\
	0.839689815203961	0.257965606342844\\
	0.839742131803304	0.257957144594197\\
	0.839842510414498	0.257941402124442\\
	0.839972823168308	0.257922568342057\\
	0.83998787885839	0.257920462529456\\
	0.840010769688855	0.257917401952518\\
	0.840023943438869	0.257915764118782\\
	0.840033942274571	0.25791461100757\\
	0.840034293354835	0.257914573501455\\
	0.840036751661524	0.257914311889591\\
	0.840037161119624	0.257914269921562\\
	0.84003801784591	0.257914186429994\\
	0.840038064679319	0.257914181925101\\
	0.840038461609938	0.257914146292685\\
	0.840038491182335	0.257914143880712\\
	0.840038517562924	0.257914141824876\\
	0.840038522707581	0.257914141459786\\
	0.84003852286414	0.257914141450283\\
	1	0.25791\\
};
\addlegendentry{$\alpha\text{ = 0.10}$};

\addplot [color=blue,dotted,line width=1.5pt]
table[row sep=crcr]{%
	1.63736483947622e-07	0.999998668847074\\
	6.22126602391041e-07	0.999995133942522\\
	7.36303920436644e-07	0.999994342230794\\
	1.21132972643008e-06	0.999992229045726\\
	1.35074328551218e-06	0.999991659528298\\
	1.92377450416938e-06	0.999989398773087\\
	3.03737577057365e-06	0.999985111206933\\
	1.17423192228156e-05	0.999959092166013\\
	4.02002737666188e-05	0.999887380263322\\
	0.000242546728205684	0.999474931093832\\
	0.00047483231604009	0.999005404487621\\
	0.000988619824415792	0.99799619331503\\
	0.00269770598539244	0.994660479110086\\
	0.0183194992540365	0.964211448523859\\
	0.100942136401185	0.80611937703414\\
	0.166968522075925	0.681072120832709\\
	0.268208599274171	0.492391674744396\\
	0.309549457537985	0.416088882877441\\
	0.317017394231043	0.402456483423726\\
	0.341678432583776	0.357463040636177\\
	0.375189303424984	0.297236183615997\\
	0.398763344230762	0.255649337242161\\
	0.408182224592285	0.239231556731663\\
	0.416519205378316	0.224828473850499\\
	0.423971974230291	0.212065186931072\\
	0.460792695321225	0.149023979969712\\
	0.464028279163374	0.143763092906885\\
	0.49224939241623	0.0986431152713933\\
	0.50607209217403	0.0795261825152272\\
	0.50862934744706	0.0760894988373681\\
	0.516900853490234	0.0650454923927826\\
	0.527063840549764	0.0536581225375052\\
	0.544105610120175	0.0378408747366536\\
	0.545295324510899	0.0367710649957762\\
	0.546037605100903	0.0361037964173105\\
	0.546992813758799	0.0352540292287575\\
	0.548072909439751	0.034339775856471\\
	0.548623464097796	0.0338751533402535\\
	0.54947260327778	0.0331752424854033\\
	0.550066975675934	0.0326955734512934\\
	0.551269697974607	0.0317384388069926\\
	0.551627328849303	0.031456797114517\\
	0.552674453502912	0.0306607761906134\\
	0.553629215156424	0.0299389946176656\\
	0.554933579719538	0.0290152388317059\\
	0.555534526813421	0.0285949633186268\\
	0.556251824390844	0.0281009217848315\\
	0.55801643299606	0.0269411847607874\\
	0.558378071529038	0.0267096850213841\\
	0.559128246895991	0.0262516777048665\\
	0.559915618290485	0.0257755275046804\\
	0.560701653504217	0.0253161658224427\\
	0.560935015231563	0.0251851386969089\\
	0.562557685185681	0.0243048589712415\\
	0.564004107664182	0.0235742443067978\\
	0.564775288161801	0.023210838500679\\
	0.565618513717895	0.0228344254130404\\
	0.566201535985078	0.0225857878883956\\
	0.566893245754242	0.0223084872982551\\
	0.566935595731575	0.0222916046539582\\
	0.567283569413782	0.0221570482321703\\
	0.56786531652953	0.0219407857404327\\
	0.568258038818306	0.0218029872422872\\
	0.56863465862402	0.0216791259247112\\
	0.568953368678731	0.0215796383854644\\
	0.569039764657516	0.0215532275699946\\
	0.569095448353502	0.0215367661073855\\
	0.569430673859726	0.0214420603155246\\
	0.569521826935221	0.0214169667352417\\
	0.569680504180617	0.0213749985326966\\
	0.569761863064456	0.021354546832378\\
	0.569802034347923	0.0213447684674602\\
	0.569830300108718	0.0213381352763942\\
	0.569945854754776	0.0213117034736127\\
	0.569993133924879	0.0213012175709412\\
	0.570049195336383	0.0212898574182637\\
	0.570102878691177	0.0212793995281797\\
	0.570108049486208	0.0212784498804611\\
	0.570128604852209	0.0212747333923257\\
	0.570143446058403	0.0212722614900582\\
	0.570147801271546	0.0212715927289049\\
	0.570150467548973	0.0212711958653442\\
	0.570153040524907	0.0212708224873842\\
	0.570156201005096	0.0212703753907768\\
	0.570157896411693	0.0212701667252837\\
	0.570158026036687	0.0212701516710295\\
	0.57015806074947	0.0212701482819448\\
	0.570158061802015	0.0212701481877926\\
	1	0.02127\\
};
\addlegendentry{$\alpha\text{ = 0.01}$};

\end{axis}
\end{tikzpicture}

%% file: fig3D_a.tex
\definecolor{mygreen}{rgb}{0,0.9,0}

\begin{tikzpicture}

\begin{axis}[%
width=0.4\textwidth,
height=0.26\textwidth,
scale only axis,
xmin=0,
xmax=1,
xtick={  0, 0.2, 0.4, 0.6, 0.8,   1},
tick align=outside,
xlabel={$R_{1}$},
xmajorgrids,
ymin=0,
ymax=1,
ytick={  0, 0.2, 0.4, 0.6, 0.8,   1},
ylabel={$R_{2}$},
ymajorgrids,
zmin=0.3,
zmax=1,
ztick={   0, 0.05,  0.1, 0.15,  0.2, 0.25,  0.3, 0.35,  0.4, 0.45,  0.5, 0.55,  0.6, 0.65,  0.7, 0.75,  0.8, 0.85,  0.9, 0.95,    1},
zticklabels={   0, ,  0.1, ,  0.2, ,  0.3, ,  0.4, ,  0.5, ,  0.6, ,  0.7, ,  0.8, ,  0.9, , 1},
zlabel={$D$},
zmajorgrids,
view={30.5}{27.6},
legend style={at={(1,1)},anchor=north east},
xlabel shift = -14pt,
ylabel shift = -14pt,
zlabel shift = -4pt,
tick label style={font=\small}, 
]

\addplot3 [color=red,only marks,mark=o,mark options={solid}]
 table[row sep=crcr] {%
1	0.88	0.4\\
0.99	0.88	0.4\\
0.96	0.88	0.41\\
0.9	0.89	0.43\\
0.71	0.9	0.49\\
0.66	0.91	0.51\\
0.56	0.92	0.55\\
0.97	0.88	0.41\\
0.93	0.89	0.42\\
0.76	0.9	0.47\\
0.7	0.91	0.5\\
0.02	1	0.8\\
0	1	0.81\\
0.98	0.88	0.4\\
0.89	0.89	0.43\\
0.88	0.89	0.43\\
0.86	0.89	0.44\\
0.81	0.89	0.45\\
0.36	0.95	0.64\\
0.05	0.99	0.79\\
0.01	1	0.81\\
0.94	0.88	0.41\\
0.49	0.93	0.59\\
0.97	0.88	0.4\\
0.91	0.89	0.42\\
0.82	0.89	0.45\\
0.23	0.96	0.7\\
1	0.87	0.4\\
0.97	0.87	0.4\\
0.95	0.88	0.41\\
1	0.86	0.4\\
0.92	0.87	0.42\\
0.04	0.99	0.79\\
1	0.84	0.4\\
0.94	0.85	0.41\\
0.88	0.86	0.43\\
1	0.82	0.4\\
0.99	0.82	0.4\\
0.98	0.82	0.4\\
0.98	0.82	0.41\\
0.97	0.82	0.41\\
0.86	0.85	0.44\\
0.81	0.86	0.45\\
0.57	0.91	0.55\\
0.32	0.95	0.66\\
0.14	0.98	0.75\\
1	0.79	0.4\\
0.99	0.79	0.4\\
0.98	0.79	0.41\\
0.97	0.79	0.41\\
0.96	0.8	0.41\\
0.95	0.8	0.41\\
0.92	0.81	0.42\\
0.89	0.82	0.43\\
0.83	0.84	0.45\\
0.7	0.88	0.5\\
0.67	0.88	0.51\\
0.37	0.94	0.64\\
0.33	0.95	0.66\\
1	0.75	0.4\\
1	0.75	0.41\\
0.99	0.75	0.41\\
0.98	0.75	0.41\\
0.98	0.76	0.41\\
0.97	0.76	0.41\\
0.96	0.76	0.41\\
0.95	0.76	0.41\\
0.94	0.77	0.42\\
0.93	0.77	0.42\\
0.84	0.81	0.45\\
0.07	0.99	0.78\\
1	0.7	0.41\\
1	0.71	0.41\\
0.99	0.71	0.41\\
0.98	0.71	0.41\\
0.97	0.72	0.41\\
0.96	0.72	0.41\\
0.96	0.72	0.42\\
0.95	0.73	0.42\\
0.94	0.73	0.42\\
0.93	0.74	0.42\\
0.82	0.79	0.46\\
0.09	0.98	0.77\\
1	0.65	0.41\\
1	0.66	0.41\\
0.99	0.66	0.41\\
0.98	0.66	0.41\\
0.97	0.67	0.42\\
0.96	0.67	0.42\\
0.96	0.68	0.42\\
0.94	0.68	0.42\\
0.94	0.69	0.42\\
0.93	0.69	0.42\\
0.93	0.69	0.43\\
0.91	0.7	0.43\\
0.9	0.71	0.44\\
0.89	0.72	0.44\\
0.88	0.72	0.44\\
0.8	0.77	0.47\\
0.74	0.8	0.49\\
0.51	0.89	0.58\\
1	0.6	0.41\\
0.99	0.6	0.42\\
0.99	0.61	0.42\\
0.98	0.61	0.42\\
0.97	0.61	0.42\\
0.96	0.62	0.42\\
0.96	0.63	0.42\\
0.95	0.63	0.42\\
0.94	0.63	0.43\\
0.93	0.64	0.43\\
0.93	0.65	0.43\\
0.92	0.65	0.43\\
0.91	0.65	0.43\\
0.91	0.66	0.44\\
0.89	0.67	0.44\\
0.88	0.68	0.44\\
0.79	0.73	0.47\\
0.73	0.77	0.49\\
0.45	0.89	0.6\\
0.26	0.94	0.69\\
0	0.99	0.81\\
1	0.54	0.42\\
0.99	0.55	0.42\\
0.98	0.55	0.42\\
0.97	0.56	0.42\\
0.96	0.57	0.43\\
0.95	0.57	0.43\\
0.94	0.58	0.43\\
0.93	0.59	0.43\\
0.92	0.59	0.43\\
0.92	0.6	0.44\\
0.91	0.6	0.44\\
0.91	0.61	0.44\\
0.9	0.61	0.44\\
0.89	0.62	0.44\\
0.87	0.63	0.45\\
0.86	0.65	0.45\\
0.84	0.66	0.46\\
0.59	0.81	0.55\\
0.46	0.87	0.6\\
0.14	0.96	0.75\\
1	0.47	0.42\\
0.99	0.48	0.43\\
0.98	0.48	0.43\\
0.98	0.49	0.43\\
0.97	0.5	0.43\\
0.96	0.5	0.43\\
0.95	0.51	0.43\\
0.94	0.52	0.44\\
0.93	0.53	0.44\\
0.92	0.53	0.44\\
0.92	0.54	0.44\\
0.91	0.55	0.44\\
0.9	0.56	0.45\\
0.89	0.57	0.45\\
0.88	0.57	0.45\\
0.87	0.58	0.45\\
0.86	0.59	0.46\\
0.85	0.6	0.46\\
0.84	0.61	0.46\\
0.58	0.79	0.55\\
0.16	0.95	0.74\\
0.08	0.97	0.77\\
1	0.4	0.43\\
0.99	0.4	0.43\\
0.99	0.41	0.43\\
0.98	0.41	0.43\\
0.97	0.42	0.43\\
0.97	0.42	0.44\\
0.97	0.43	0.44\\
0.95	0.44	0.44\\
0.94	0.45	0.44\\
0.92	0.47	0.44\\
0.91	0.48	0.45\\
0.9	0.49	0.45\\
0.9	0.5	0.45\\
0.89	0.51	0.45\\
0.88	0.51	0.46\\
0.87	0.52	0.46\\
0.86	0.53	0.46\\
0.85	0.54	0.46\\
0.84	0.56	0.47\\
0.82	0.57	0.47\\
0.81	0.59	0.48\\
0.79	0.6	0.48\\
0.69	0.68	0.51\\
0.18	0.94	0.73\\
0.02	0.98	0.8\\
0.01	0.98	0.81\\
0	0.98	0.81\\
1	0.32	0.44\\
0.99	0.33	0.44\\
0.98	0.33	0.44\\
0.98	0.34	0.44\\
0.97	0.35	0.44\\
0.96	0.36	0.44\\
0.95	0.36	0.44\\
0.95	0.37	0.45\\
0.94	0.38	0.45\\
0.94	0.39	0.45\\
0.93	0.39	0.45\\
0.91	0.42	0.46\\
0.9	0.43	0.46\\
0.89	0.44	0.46\\
0.88	0.45	0.46\\
0.87	0.46	0.46\\
0.86	0.47	0.47\\
0.85	0.48	0.47\\
0.84	0.49	0.47\\
0.83	0.51	0.47\\
0.76	0.58	0.49\\
0.62	0.7	0.54\\
0.02	0.97	0.8\\
0.01	0.97	0.81\\
1	0.23	0.45\\
0.99	0.23	0.45\\
0.99	0.24	0.45\\
0.98	0.25	0.45\\
0.97	0.26	0.45\\
0.97	0.27	0.45\\
0.96	0.28	0.45\\
0.95	0.29	0.45\\
0.94	0.3	0.45\\
0.93	0.32	0.46\\
0.91	0.34	0.46\\
0.9	0.36	0.46\\
0.89	0.37	0.46\\
0.88	0.39	0.47\\
0.87	0.4	0.47\\
0.86	0.41	0.47\\
0.84	0.43	0.48\\
0.83	0.44	0.48\\
0.81	0.47	0.49\\
0.79	0.49	0.49\\
0.64	0.65	0.54\\
0.58	0.7	0.56\\
0.36	0.83	0.65\\
0.01	0.96	0.81\\
0	0.97	0.81\\
1	0.14	0.45\\
1	0.14	0.46\\
0.99	0.14	0.46\\
0.99	0.15	0.46\\
0.98	0.16	0.46\\
0.98	0.17	0.46\\
0.97	0.17	0.46\\
0.96	0.18	0.46\\
0.95	0.2	0.46\\
0.94	0.21	0.46\\
0.93	0.23	0.46\\
0.92	0.24	0.47\\
0.92	0.25	0.47\\
0.91	0.26	0.47\\
0.9	0.27	0.47\\
0.89	0.3	0.47\\
0.87	0.32	0.48\\
0.86	0.34	0.48\\
0.85	0.35	0.48\\
0.84	0.37	0.49\\
0.82	0.39	0.49\\
0.81	0.41	0.49\\
0.8	0.42	0.5\\
0.78	0.45	0.5\\
0.76	0.47	0.5\\
0.73	0.5	0.52\\
0.02	0.95	0.81\\
0	0.96	0.81\\
0	0.96	0.82\\
1	0.04	0.47\\
1	0.01	0.47\\
1	0.03	0.47\\
1	0.02	0.47\\
0.99	0.04	0.47\\
0.99	0.05	0.47\\
0.98	0.03	0.47\\
0.98	0.06	0.47\\
0.98	0.07	0.47\\
0.97	0.07	0.47\\
0.97	0.08	0.47\\
0.96	0.09	0.47\\
0.96	0.1	0.47\\
0.95	0.11	0.47\\
0.94	0.13	0.47\\
0.93	0.15	0.47\\
0.92	0.16	0.48\\
0.91	0.17	0.48\\
0.91	0.18	0.48\\
0.9	0.19	0.48\\
0.88	0.22	0.48\\
0.87	0.25	0.49\\
0.85	0.27	0.49\\
0.84	0.29	0.49\\
0.83	0.3	0.49\\
0.82	0.33	0.5\\
0.8	0.35	0.5\\
0.79	0.37	0.5\\
0.77	0.4	0.51\\
0.75	0.42	0.51\\
0.73	0.45	0.52\\
0.71	0.48	0.53\\
0.64	0.56	0.55\\
0.57	0.63	0.58\\
0.37	0.78	0.65\\
0.3	0.82	0.68\\
0.13	0.9	0.76\\
0.01	0.94	0.81\\
0	0.95	0.82\\
1	0	0.47\\
0.99	0	0.47\\
0.98	0	0.47\\
0.97	0	0.48\\
0.96	0	0.48\\
0.96	0.01	0.48\\
0.95	0.01	0.48\\
0.95	0.02	0.48\\
0.94	0.03	0.48\\
0.94	0.04	0.48\\
0.93	0.04	0.49\\
0.93	0.06	0.49\\
0.92	0.07	0.49\\
0.91	0.08	0.49\\
0.91	0.09	0.49\\
0.9	0.11	0.49\\
0.89	0.12	0.49\\
0.88	0.14	0.49\\
0.87	0.15	0.49\\
0.86	0.17	0.5\\
0.85	0.19	0.5\\
0.84	0.21	0.5\\
0.83	0.23	0.5\\
0.81	0.26	0.51\\
0.8	0.29	0.51\\
0.78	0.31	0.51\\
0.76	0.33	0.52\\
0.74	0.38	0.53\\
0.72	0.41	0.53\\
0.26	0.82	0.7\\
0.01	0.93	0.81\\
0	0.93	0.82\\
0.95	0	0.48\\
0.94	0	0.49\\
0.93	0	0.49\\
0.92	0	0.49\\
0.91	0.01	0.5\\
0.9	0.01	0.5\\
0.9	0.02	0.5\\
0.89	0.02	0.5\\
0.88	0.05	0.5\\
0.87	0.06	0.5\\
0.85	0.11	0.51\\
0.84	0.12	0.51\\
0.82	0.16	0.51\\
0.81	0.19	0.52\\
0.79	0.22	0.52\\
0.77	0.25	0.52\\
0.75	0.29	0.53\\
0.73	0.31	0.53\\
0.7	0.37	0.54\\
0.67	0.41	0.55\\
0.64	0.46	0.56\\
0.52	0.58	0.6\\
0.32	0.76	0.68\\
0.06	0.9	0.79\\
0.02	0.91	0.81\\
0	0.92	0.82\\
0.92	0	0.5\\
0.91	0	0.5\\
0.9	0	0.5\\
0.89	0	0.5\\
0.89	0	0.51\\
0.88	0	0.51\\
0.87	0	0.51\\
0.86	0.01	0.52\\
0.85	0.02	0.52\\
0.82	0.07	0.52\\
0.81	0.1	0.53\\
0.79	0.12	0.53\\
0.77	0.18	0.53\\
0.75	0.22	0.54\\
0.72	0.27	0.55\\
0.69	0.32	0.55\\
0.65	0.38	0.56\\
0.05	0.88	0.8\\
0.01	0.9	0.82\\
0	0.91	0.82\\
0.86	0	0.52\\
0.85	0	0.52\\
0.84	0	0.53\\
0.83	0	0.53\\
0.82	0.01	0.53\\
0.81	0.02	0.54\\
0.79	0.05	0.54\\
0.77	0.1	0.54\\
0.74	0.15	0.55\\
0.71	0.21	0.56\\
0.68	0.28	0.57\\
0.6	0.4	0.59\\
0.37	0.65	0.67\\
0.25	0.74	0.71\\
0.02	0.88	0.81\\
0	0.89	0.82\\
0.82	0	0.53\\
0.81	0	0.54\\
0.8	0	0.54\\
0.79	0	0.55\\
0.78	0	0.55\\
0.77	0.02	0.55\\
0.75	0.06	0.56\\
0.67	0.22	0.58\\
0.63	0.28	0.59\\
0.47	0.51	0.64\\
0.06	0.84	0.8\\
0.01	0.87	0.82\\
0	0.87	0.82\\
0	0.87	0.83\\
0.77	0	0.55\\
0.76	0	0.56\\
0.75	0	0.56\\
0.74	0.01	0.57\\
0.72	0.05	0.58\\
0.66	0.16	0.59\\
0.14	0.77	0.77\\
0.02	0.84	0.82\\
0	0.85	0.83\\
0.74	0	0.57\\
0.73	0	0.58\\
0.71	0	0.59\\
0.67	0.07	0.6\\
0.1	0.77	0.79\\
0	0.83	0.83\\
0.72	0	0.58\\
0.68	0	0.6\\
0.66	0.01	0.61\\
0	0.81	0.83\\
0.69	0	0.59\\
0.61	0.03	0.63\\
0	0.78	0.84\\
0	0.79	0.84\\
0	0.76	0.84\\
0.67	0	0.61\\
0	0.73	0.84\\
0	0.73	0.85\\
0.65	0	0.62\\
0	0.7	0.85\\
0	0.71	0.85\\
0.64	0	0.62\\
0.62	0	0.63\\
0	0.68	0.85\\
0	0.65	0.86\\
0	0.61	0.87\\
0	0.58	0.87\\
0	0.54	0.88\\
0	0.51	0.89\\
0	0.47	0.89\\
0	0.43	0.9\\
0	0.39	0.91\\
0	0.34	0.92\\
0	0.35	0.92\\
0	0.3	0.93\\
0	0.25	0.94\\
-0	0.2	0.95\\
0	0.16	0.96\\
-0	0.15	0.96\\
0	0.11	0.97\\
-0	0.09	0.98\\
0	0.1	0.97\\
0	0.05	0.99\\
0	0.01	1\\
0	0	1\\
0.61	0	0.64\\
0.33	0	0.8\\
0.1	-0	0.94\\
0.77	-0	0.56\\
0.48	0	0.71\\
0.02	0	0.98\\
0.01	0	1\\
0.73	0.9	0.48\\
0.58	0.92	0.54\\
0.79	0.9	0.46\\
0.68	0.91	0.5\\
0.84	0.89	0.45\\
0.38	0.94	0.64\\
0.81	0.85	0.46\\
0.71	0.78	0.5\\
0.65	0.72	0.53\\
0.92	0.7	0.43\\
0.86	0.69	0.45\\
0.74	0.73	0.49\\
0.82	0.52	0.48\\
0.75	0.59	0.5\\
0.78	0.51	0.5\\
0.69	0.44	0.54\\
0.6	0.27	0.6\\
0.26	0.61	0.73\\
0.54	0	0.68\\
0.41	-0	0.75\\
0.38	-0	0.76\\
0.3	-0	0.81\\
0.27	0	0.83\\
0.2	-0	0.87\\
0.24	-0	0.85\\
0.17	-0	0.89\\
0.75	0	0.57\\
0.13	-0	0.92\\
0	0.06	0.99\\
0.06	-0	0.96\\
};
 \addlegendentry{$\mathcal{RD}_{\mathrm{CEO}}^{1}$};

\addplot3 [color=mygreen,only marks,mark=o,mark options={solid}]
 table[row sep=crcr] {%
0.88	1	0.4\\
0.87	1	0.4\\
0.86	1	0.4\\
0.84	1	0.41\\
0.69	1	0.46\\
0.62	1	0.49\\
0.51	1	0.54\\
0.49	1	0.55\\
0.45	1	0.57\\
0.37	1	0.61\\
0.3	1	0.64\\
0.28	1	0.66\\
0.25	1	0.67\\
0.19	1	0.71\\
0.15	1	0.73\\
0	1	0.81\\
0.82	1	0.42\\
0.8	1	0.42\\
0.79	1	0.42\\
0.76	1	0.43\\
0.74	1	0.44\\
0.72	1	0.45\\
0.68	1	0.47\\
0.66	1	0.47\\
0.65	1	0.48\\
0.52	1	0.54\\
0.47	1	0.56\\
0.43	1	0.58\\
0.35	1	0.62\\
0.03	1	0.8\\
0.83	1	0.41\\
0.79	1	0.43\\
0.58	1	0.51\\
0.42	1	0.59\\
0.39	1	0.6\\
0.33	1	0.63\\
0.12	1	0.75\\
0.07	1	0.77\\
0.01	1	0.81\\
0.85	1	0.41\\
0.77	1	0.43\\
0.82	1	0.41\\
0.06	1	0.78\\
0.03	1	0.79\\
0.88	0.99	0.4\\
0.67	1	0.47\\
0.4	1	0.59\\
0.88	0.96	0.4\\
0.87	0.99	0.4\\
0.38	1	0.61\\
0.88	0.93	0.4\\
0.88	0.95	0.4\\
0.88	0.97	0.4\\
0.79	0.99	0.42\\
0.73	0.99	0.44\\
0.23	1	0.68\\
0.89	0.9	0.4\\
0.88	0.94	0.4\\
0.87	0.96	0.4\\
0.87	0.97	0.4\\
0.86	0.97	0.4\\
0.59	1	0.5\\
0.89	0.88	0.4\\
0.87	0.95	0.4\\
0.86	0.95	0.4\\
0.86	0.95	0.41\\
0.85	0.96	0.41\\
0.67	0.99	0.47\\
0.13	1	0.74\\
0.9	0.78	0.41\\
0.89	0.82	0.41\\
0.88	0.9	0.4\\
0.88	0.91	0.4\\
0.88	0.92	0.4\\
0.87	0.92	0.4\\
0.87	0.92	0.41\\
0.86	0.93	0.41\\
0.1	1	0.76\\
0.9	0.75	0.41\\
0.89	0.86	0.41\\
0.88	0.88	0.41\\
0.87	0.89	0.41\\
0.87	0.9	0.41\\
0.86	0.91	0.41\\
0.85	0.91	0.41\\
0.84	0.92	0.41\\
0.76	0.95	0.44\\
0.92	0.58	0.42\\
0.91	0.62	0.42\\
0.91	0.69	0.41\\
0.9	0.72	0.41\\
0.89	0.83	0.41\\
0.88	0.85	0.41\\
0.87	0.87	0.41\\
0.86	0.87	0.41\\
0.86	0.88	0.41\\
0.85	0.88	0.41\\
0.85	0.89	0.41\\
0.82	0.91	0.42\\
0.79	0.92	0.43\\
0.71	0.95	0.45\\
0.69	0.95	0.46\\
0.53	0.97	0.53\\
0.23	0.99	0.68\\
0.06	0.99	0.78\\
0.02	0.99	0.8\\
0	0.99	0.81\\
0.93	0.52	0.43\\
0.92	0.53	0.43\\
0.91	0.61	0.42\\
0.91	0.67	0.42\\
0.9	0.76	0.41\\
0.89	0.78	0.41\\
0.89	0.79	0.41\\
0.89	0.8	0.41\\
0.88	0.82	0.41\\
0.87	0.83	0.41\\
0.87	0.84	0.41\\
0.86	0.84	0.41\\
0.86	0.85	0.41\\
0.85	0.86	0.42\\
0.84	0.87	0.42\\
0.83	0.87	0.42\\
0.82	0.88	0.42\\
0.94	0.43	0.43\\
0.9	0.74	0.41\\
0.88	0.77	0.41\\
0.86	0.81	0.42\\
0.85	0.82	0.42\\
0.85	0.83	0.42\\
0.84	0.83	0.42\\
0.84	0.84	0.42\\
0.83	0.84	0.42\\
0.82	0.85	0.43\\
0.79	0.87	0.43\\
0.78	0.88	0.44\\
0.76	0.89	0.44\\
0.72	0.9	0.46\\
0.4	0.96	0.59\\
0.03	0.99	0.79\\
0.01	0.99	0.81\\
0.97	0.21	0.45\\
0.93	0.49	0.43\\
0.92	0.54	0.42\\
0.92	0.57	0.42\\
0.9	0.69	0.42\\
0.89	0.72	0.42\\
0.89	0.73	0.42\\
0.87	0.75	0.42\\
0.87	0.76	0.42\\
0.86	0.78	0.42\\
0.85	0.79	0.42\\
0.84	0.8	0.42\\
0.83	0.81	0.43\\
0.82	0.82	0.43\\
0.79	0.84	0.43\\
0.7	0.88	0.46\\
0.66	0.9	0.48\\
0.52	0.93	0.54\\
0.31	0.96	0.64\\
0.03	0.98	0.79\\
0.01	0.98	0.81\\
0	0.98	0.81\\
1	0.02	0.47\\
0.99	0.04	0.47\\
0.99	0.07	0.46\\
0.98	0.12	0.46\\
0.98	0.15	0.46\\
0.96	0.25	0.45\\
0.95	0.33	0.44\\
0.94	0.37	0.44\\
0.92	0.55	0.42\\
0.92	0.56	0.42\\
0.91	0.6	0.42\\
0.9	0.66	0.42\\
0.89	0.67	0.42\\
0.89	0.69	0.42\\
0.88	0.7	0.42\\
0.87	0.72	0.42\\
0.86	0.73	0.42\\
0.86	0.74	0.42\\
0.85	0.75	0.42\\
0.84	0.76	0.43\\
0.83	0.77	0.43\\
0.83	0.78	0.43\\
0.82	0.78	0.43\\
0.81	0.79	0.43\\
0.8	0.8	0.44\\
0.78	0.82	0.44\\
0.73	0.85	0.46\\
0.59	0.9	0.51\\
0.45	0.93	0.57\\
0.09	0.97	0.76\\
0.03	0.97	0.79\\
0.01	0.97	0.81\\
0	0.97	0.81\\
1	0	0.47\\
0.99	0.08	0.46\\
0.96	0.29	0.44\\
0.95	0.35	0.44\\
0.93	0.46	0.43\\
0.92	0.52	0.43\\
0.92	0.55	0.43\\
0.91	0.59	0.42\\
0.9	0.6	0.42\\
0.9	0.61	0.42\\
0.89	0.63	0.42\\
0.89	0.64	0.42\\
0.87	0.68	0.42\\
0.86	0.7	0.43\\
0.84	0.72	0.43\\
0.83	0.73	0.43\\
0.83	0.74	0.43\\
0.82	0.74	0.43\\
0.81	0.75	0.43\\
0.81	0.76	0.44\\
0.8	0.76	0.44\\
0.78	0.78	0.44\\
0.7	0.83	0.47\\
0.65	0.85	0.49\\
0.47	0.91	0.57\\
0.13	0.96	0.74\\
0.07	0.96	0.78\\
0.02	0.97	0.8\\
1	0.03	0.47\\
0.99	0.05	0.46\\
0.98	0.1	0.46\\
0.97	0.18	0.45\\
0.96	0.27	0.45\\
0.94	0.39	0.44\\
0.94	0.42	0.43\\
0.92	0.5	0.43\\
0.91	0.53	0.43\\
0.9	0.56	0.43\\
0.89	0.59	0.43\\
0.89	0.6	0.43\\
0.87	0.64	0.43\\
0.86	0.65	0.43\\
0.85	0.66	0.43\\
0.82	0.71	0.44\\
0.81	0.72	0.44\\
0.8	0.72	0.44\\
0.79	0.73	0.44\\
0.78	0.74	0.45\\
0.76	0.76	0.45\\
0.74	0.78	0.46\\
0.7	0.8	0.47\\
0.67	0.82	0.48\\
0.44	0.9	0.58\\
0.02	0.96	0.8\\
0.01	0.96	0.81\\
0	0.96	0.81\\
0.97	0.19	0.45\\
0.97	0.22	0.45\\
0.94	0.4	0.44\\
0.93	0.41	0.44\\
0.93	0.43	0.43\\
0.92	0.45	0.43\\
0.91	0.49	0.43\\
0.91	0.5	0.43\\
0.89	0.54	0.43\\
0.88	0.56	0.43\\
0.88	0.57	0.43\\
0.86	0.61	0.43\\
0.85	0.62	0.44\\
0.84	0.63	0.44\\
0.83	0.65	0.44\\
0.82	0.66	0.44\\
0.81	0.67	0.44\\
0.8	0.68	0.45\\
0.79	0.7	0.45\\
0.77	0.71	0.45\\
0.74	0.73	0.46\\
0.73	0.75	0.47\\
0.64	0.8	0.5\\
0.33	0.9	0.64\\
0.18	0.92	0.71\\
0.04	0.94	0.79\\
0	0.95	0.81\\
0	0.95	0.82\\
0.95	0.3	0.44\\
0.94	0.34	0.44\\
0.92	0.42	0.44\\
0.92	0.43	0.44\\
0.9	0.48	0.44\\
0.88	0.51	0.44\\
0.88	0.52	0.44\\
0.87	0.54	0.44\\
0.86	0.55	0.44\\
0.85	0.57	0.44\\
0.84	0.59	0.44\\
0.83	0.6	0.44\\
0.82	0.62	0.45\\
0.81	0.63	0.45\\
0.8	0.64	0.45\\
0.79	0.66	0.45\\
0.78	0.66	0.46\\
0.77	0.67	0.46\\
0.7	0.73	0.48\\
0.56	0.81	0.53\\
0.04	0.93	0.79\\
0.01	0.93	0.81\\
0	0.93	0.82\\
1	0.01	0.47\\
0.98	0.11	0.46\\
0.97	0.15	0.46\\
0.97	0.16	0.46\\
0.97	0.2	0.45\\
0.96	0.21	0.45\\
0.94	0.31	0.45\\
0.93	0.35	0.44\\
0.92	0.36	0.44\\
0.92	0.37	0.44\\
0.92	0.39	0.44\\
0.91	0.41	0.44\\
0.9	0.43	0.44\\
0.87	0.49	0.44\\
0.87	0.5	0.44\\
0.86	0.51	0.44\\
0.86	0.52	0.45\\
0.83	0.55	0.45\\
0.82	0.57	0.45\\
0.8	0.6	0.46\\
0.79	0.61	0.46\\
0.78	0.62	0.46\\
0.77	0.63	0.46\\
0.76	0.64	0.47\\
0.75	0.65	0.47\\
0.74	0.66	0.47\\
0.72	0.68	0.48\\
0.69	0.7	0.49\\
0.66	0.72	0.5\\
0.51	0.8	0.56\\
0.25	0.88	0.68\\
0.16	0.9	0.73\\
0.09	0.91	0.77\\
0.04	0.92	0.79\\
0.01	0.92	0.81\\
0.01	0.92	0.82\\
0	0.92	0.82\\
0.99	0.03	0.47\\
0.99	0.05	0.47\\
0.99	0.06	0.46\\
0.97	0.14	0.46\\
0.94	0.26	0.45\\
0.93	0.3	0.45\\
0.92	0.32	0.45\\
0.91	0.34	0.45\\
0.91	0.36	0.45\\
0.89	0.41	0.45\\
0.86	0.46	0.45\\
0.84	0.5	0.45\\
0.82	0.52	0.46\\
0.81	0.53	0.46\\
0.81	0.54	0.46\\
0.78	0.58	0.47\\
0.77	0.59	0.47\\
0.76	0.6	0.47\\
0.75	0.61	0.47\\
0.74	0.62	0.48\\
0.73	0.63	0.48\\
0.67	0.68	0.5\\
0.66	0.68	0.5\\
0.59	0.73	0.53\\
0.39	0.82	0.62\\
0.35	0.83	0.64\\
0.26	0.85	0.68\\
0.06	0.9	0.79\\
0.01	0.9	0.82\\
0	0.91	0.82\\
0.99	0.01	0.47\\
0.98	0.06	0.47\\
0.98	0.08	0.46\\
0.97	0.1	0.46\\
0.96	0.17	0.46\\
0.95	0.19	0.46\\
0.92	0.27	0.45\\
0.9	0.33	0.45\\
0.87	0.39	0.46\\
0.87	0.4	0.46\\
0.85	0.43	0.46\\
0.84	0.44	0.46\\
0.83	0.47	0.46\\
0.8	0.51	0.47\\
0.79	0.52	0.47\\
0.78	0.53	0.47\\
0.76	0.55	0.48\\
0.75	0.56	0.48\\
0.74	0.57	0.48\\
0.72	0.6	0.49\\
0.67	0.64	0.5\\
0.49	0.75	0.57\\
0.45	0.77	0.59\\
0.23	0.84	0.7\\
0.13	0.86	0.75\\
0.08	0.87	0.78\\
0.03	0.88	0.81\\
0.01	0.89	0.82\\
0	0.89	0.82\\
0.99	0	0.47\\
0.98	0.03	0.47\\
0.98	0.04	0.47\\
0.97	0.06	0.47\\
0.97	0.08	0.47\\
0.95	0.16	0.46\\
0.94	0.17	0.46\\
0.93	0.21	0.46\\
0.92	0.22	0.46\\
0.9	0.27	0.46\\
0.89	0.31	0.46\\
0.88	0.32	0.46\\
0.86	0.35	0.46\\
0.85	0.38	0.47\\
0.84	0.4	0.47\\
0.82	0.42	0.47\\
0.79	0.47	0.47\\
0.78	0.48	0.48\\
0.77	0.49	0.48\\
0.76	0.5	0.48\\
0.75	0.52	0.49\\
0.74	0.53	0.49\\
0.73	0.54	0.49\\
0.72	0.55	0.49\\
0.69	0.58	0.5\\
0.61	0.64	0.53\\
0.5	0.71	0.57\\
0.31	0.79	0.66\\
0.21	0.82	0.71\\
0.14	0.84	0.75\\
0.03	0.86	0.81\\
0.01	0.87	0.82\\
0	0.87	0.82\\
0	0.87	0.83\\
0.98	0	0.47\\
0.98	0.01	0.47\\
0.98	0.02	0.47\\
0.97	0.02	0.47\\
0.97	0.05	0.47\\
0.96	0.06	0.47\\
0.96	0.08	0.47\\
0.93	0.16	0.47\\
0.9	0.22	0.47\\
0.89	0.25	0.47\\
0.88	0.27	0.47\\
0.85	0.33	0.47\\
0.82	0.38	0.48\\
0.81	0.4	0.48\\
0.8	0.41	0.48\\
0.79	0.42	0.48\\
0.77	0.45	0.49\\
0.75	0.47	0.49\\
0.74	0.48	0.49\\
0.73	0.5	0.5\\
0.72	0.51	0.5\\
0.69	0.53	0.51\\
0.65	0.57	0.52\\
0.29	0.77	0.67\\
0.04	0.84	0.8\\
0	0.85	0.83\\
0.97	0	0.48\\
0.97	0.01	0.48\\
0.96	0.01	0.48\\
0.96	0.03	0.48\\
0.95	0.04	0.48\\
0.95	0.06	0.48\\
0.94	0.08	0.48\\
0.93	0.1	0.47\\
0.92	0.12	0.47\\
0.91	0.14	0.47\\
0.91	0.16	0.47\\
0.9	0.18	0.47\\
0.89	0.2	0.48\\
0.87	0.23	0.48\\
0.86	0.25	0.48\\
0.84	0.29	0.48\\
0.82	0.32	0.48\\
0.81	0.34	0.48\\
0.79	0.37	0.49\\
0.77	0.4	0.49\\
0.74	0.43	0.5\\
0.73	0.45	0.5\\
0.72	0.46	0.51\\
0.71	0.47	0.51\\
0.69	0.49	0.52\\
0.68	0.5	0.52\\
0.66	0.52	0.52\\
0.65	0.53	0.53\\
0.6	0.57	0.55\\
0.54	0.61	0.57\\
0.46	0.66	0.6\\
0.08	0.81	0.79\\
0.01	0.83	0.83\\
0	0.83	0.83\\
0.96	0	0.48\\
0.95	0	0.48\\
0.95	0.01	0.48\\
0.94	0.02	0.48\\
0.94	0.04	0.48\\
0.93	0.06	0.48\\
0.92	0.08	0.48\\
0.91	0.1	0.48\\
0.89	0.14	0.48\\
0.87	0.18	0.48\\
0.86	0.2	0.49\\
0.84	0.23	0.49\\
0.82	0.27	0.49\\
0.81	0.29	0.49\\
0.8	0.31	0.5\\
0.78	0.33	0.5\\
0.77	0.34	0.5\\
0.75	0.38	0.51\\
0.73	0.39	0.51\\
0.71	0.43	0.52\\
0.69	0.44	0.52\\
0.68	0.46	0.53\\
0.66	0.47	0.53\\
0.63	0.5	0.54\\
0.56	0.56	0.57\\
0.43	0.64	0.62\\
0.22	0.74	0.72\\
0.19	0.75	0.74\\
0.04	0.8	0.81\\
0.01	0.81	0.83\\
0	0.81	0.83\\
0.94	0	0.49\\
0.93	0	0.49\\
0.93	0.01	0.49\\
0.92	0.03	0.49\\
0.91	0.05	0.49\\
0.9	0.06	0.49\\
0.87	0.12	0.49\\
0.85	0.17	0.49\\
0.84	0.19	0.5\\
0.82	0.22	0.5\\
0.79	0.27	0.5\\
0.78	0.28	0.51\\
0.75	0.32	0.51\\
0.74	0.34	0.52\\
0.72	0.35	0.52\\
0.71	0.37	0.52\\
0.69	0.39	0.53\\
0.68	0.4	0.53\\
0.66	0.42	0.54\\
0.64	0.44	0.54\\
0.63	0.46	0.55\\
0.57	0.5	0.57\\
0.27	0.69	0.7\\
0.16	0.73	0.75\\
0.08	0.76	0.8\\
0.04	0.77	0.82\\
0.01	0.78	0.83\\
0	0.78	0.84\\
0	0.79	0.84\\
0.92	0	0.49\\
0.91	0.01	0.5\\
0.89	0.03	0.5\\
0.88	0.05	0.5\\
0.87	0.07	0.5\\
0.86	0.1	0.5\\
0.84	0.14	0.51\\
0.81	0.18	0.51\\
0.8	0.2	0.51\\
0.78	0.22	0.51\\
0.77	0.24	0.52\\
0.76	0.26	0.52\\
0.73	0.3	0.53\\
0.69	0.34	0.54\\
0.68	0.35	0.54\\
0.66	0.38	0.55\\
0.64	0.39	0.55\\
0.63	0.41	0.56\\
0.48	0.53	0.61\\
0.35	0.61	0.67\\
0.22	0.68	0.73\\
0.09	0.73	0.79\\
0.05	0.74	0.82\\
0.02	0.75	0.83\\
0	0.76	0.84\\
0.92	0	0.5\\
0.89	0	0.5\\
0.89	0	0.51\\
0.88	0.01	0.51\\
0.87	0.02	0.51\\
0.86	0.05	0.51\\
0.84	0.07	0.51\\
0.82	0.11	0.52\\
0.8	0.13	0.52\\
0.79	0.16	0.52\\
0.78	0.17	0.52\\
0.75	0.22	0.53\\
0.73	0.24	0.53\\
0.71	0.26	0.54\\
0.7	0.28	0.54\\
0.68	0.3	0.55\\
0.66	0.32	0.55\\
0.64	0.34	0.56\\
0.63	0.36	0.56\\
0.61	0.38	0.57\\
0.57	0.42	0.58\\
0.03	0.72	0.83\\
0.01	0.73	0.84\\
0	0.73	0.84\\
0	0.73	0.85\\
0.88	0	0.51\\
0.81	0.07	0.53\\
0.8	0.09	0.53\\
0.78	0.11	0.53\\
0.77	0.14	0.53\\
0.75	0.16	0.54\\
0.72	0.2	0.55\\
0.7	0.23	0.55\\
0.69	0.24	0.56\\
0.67	0.26	0.56\\
0.65	0.29	0.57\\
0.63	0.31	0.57\\
0.61	0.33	0.58\\
0.59	0.35	0.59\\
0.56	0.38	0.59\\
0.42	0.49	0.65\\
0.39	0.51	0.66\\
0.28	0.58	0.71\\
0.16	0.64	0.77\\
0.01	0.7	0.84\\
0	0.7	0.85\\
0	0.71	0.85\\
0.79	0.04	0.54\\
0.78	0.07	0.54\\
0.76	0.09	0.55\\
0.74	0.11	0.55\\
0.72	0.14	0.55\\
0.71	0.17	0.56\\
0.69	0.18	0.56\\
0.67	0.21	0.57\\
0.65	0.24	0.57\\
0.63	0.26	0.58\\
0.61	0.28	0.59\\
0.58	0.31	0.59\\
0.56	0.33	0.6\\
0.52	0.37	0.62\\
0.46	0.41	0.64\\
0.06	0.65	0.83\\
0.01	0.67	0.85\\
0	0.68	0.85\\
0.79	0	0.55\\
0.78	0.01	0.55\\
0.77	0.02	0.55\\
0.69	0.13	0.57\\
0.67	0.15	0.58\\
0.65	0.18	0.58\\
0.63	0.21	0.59\\
0.61	0.22	0.6\\
0.59	0.25	0.6\\
0.56	0.28	0.61\\
0.51	0.32	0.63\\
0.46	0.37	0.65\\
0.4	0.42	0.67\\
0.33	0.47	0.7\\
0.04	0.63	0.84\\
0.01	0.64	0.86\\
0	0.64	0.86\\
0	0.65	0.86\\
0.78	0	0.55\\
0.77	0	0.55\\
0.76	0	0.56\\
0.63	0.15	0.6\\
0.61	0.16	0.6\\
0.59	0.2	0.61\\
0.56	0.23	0.62\\
0.53	0.25	0.63\\
0.51	0.27	0.64\\
0.48	0.3	0.65\\
0.45	0.33	0.66\\
0.39	0.38	0.68\\
0.24	0.48	0.75\\
0.19	0.51	0.77\\
0.04	0.6	0.85\\
0	0.61	0.86\\
0	0.61	0.87\\
0.62	0.11	0.61\\
0.59	0.14	0.62\\
0.56	0.17	0.63\\
0.54	0.19	0.64\\
0.5	0.23	0.65\\
0.48	0.25	0.66\\
0.45	0.28	0.67\\
0.31	0.39	0.73\\
0.03	0.56	0.86\\
0	0.58	0.87\\
0.59	0.09	0.63\\
0.54	0.15	0.65\\
0.51	0.18	0.66\\
0.47	0.21	0.67\\
0.41	0.26	0.69\\
0.26	0.39	0.76\\
0.05	0.52	0.85\\
0.01	0.54	0.87\\
0	0.54	0.88\\
0.47	0.16	0.68\\
0.44	0.18	0.69\\
0.4	0.23	0.71\\
0.36	0.27	0.72\\
0.19	0.39	0.79\\
0.14	0.42	0.82\\
0.09	0.46	0.84\\
0.01	0.5	0.88\\
0	0.51	0.89\\
0.47	0.1	0.69\\
0.41	0.16	0.72\\
0.35	0.22	0.74\\
0.23	0.32	0.79\\
0.02	0.46	0.88\\
0	0.47	0.89\\
0.36	0.15	0.75\\
0.31	0.2	0.76\\
0.1	0.36	0.85\\
0.02	0.42	0.89\\
0	0.43	0.9\\
0.2	0.24	0.82\\
0.04	0.36	0.89\\
0	0.39	0.91\\
0.25	0.15	0.81\\
0.17	0.22	0.84\\
0.05	0.31	0.9\\
0.02	0.33	0.91\\
0	0.34	0.92\\
0.17	0.16	0.85\\
0.1	0.22	0.88\\
0	0.3	0.93\\
0.04	0.22	0.92\\
0	0.25	0.94\\
0	0.2	0.95\\
0	0.21	0.95\\
0	0.15	0.96\\
0	0.16	0.96\\
-0	0.09	0.98\\
0	0.11	0.97\\
0	0.1	0.98\\
0.03	0.02	0.98\\
0	0.06	0.99\\
-0	0.04	0.99\\
0.02	0	0.99\\
0	0.01	1\\
0.09	0	0.94\\
0	0	1\\
0.8	0	0.54\\
0.02	-0	0.98\\
0.48	0	0.71\\
0.4	-0	0.75\\
0.45	-0	0.73\\
0.01	-0	1\\
0.74	-0	0.57\\
0.49	0	0.7\\
0.69	0	0.59\\
0.2	0	0.87\\
0.86	0	0.52\\
0.06	-0	0.96\\
0.17	-0	0.89\\
0.71	0	0.59\\
0.65	-0	0.62\\
0.62	-0	0.63\\
0.68	-0	0.6\\
0.13	0	0.92\\
0.75	-0	0.57\\
0.84	-0	0.53\\
0.82	-0	0.53\\
0.3	-0	0.81\\
0.87	0	0.51\\
0.52	-0	0.69\\
0.5	0	0.7\\
0.9	0	0.5\\
0.91	0	0.5\\
0.85	-0	0.52\\
0.59	-0	0.65\\
0.23	-0	0.86\\
0.73	-0	0.58\\
0.47	-0	0.71\\
0.53	-0	0.68\\
0.51	-0	0.69\\
0.57	-0	0.66\\
0.72	-0	0.58\\
0.27	0	0.83\\
0.55	-0	0.67\\
0.65	0	0.61\\
0.01	-0	0.99\\
0.61	-0	0.64\\
0.67	-0	0.61\\
0.36	0	0.78\\
0.43	-0	0.74\\
0.88	0.98	0.4\\
0.22	1	0.69\\
0.81	1	0.42\\
0.52	0.99	0.53\\
0.78	0.99	0.43\\
0.89	0.89	0.4\\
0.65	0.99	0.48\\
0.89	0.84	0.41\\
0.88	0.87	0.41\\
0.89	0.75	0.41\\
0.78	0.84	0.44\\
0.77	0.85	0.44\\
0.34	0.95	0.63\\
0.94	0.38	0.44\\
0.9	0.62	0.42\\
0.79	0.77	0.44\\
0.77	0.79	0.45\\
0.89	0.49	0.44\\
0.76	0.68	0.46\\
0.72	0.64	0.48\\
0.7	0.52	0.5\\
0.61	0.43	0.56\\
0.54	0.4	0.6\\
0.24	0.6	0.73\\
0.22	0.57	0.75\\
0.15	0.57	0.78\\
0.75	0.05	0.56\\
0.27	0.42	0.74\\
0.56	0.12	0.64\\
0.24	-0	0.85\\
0	0.1	0.97\\
0.77	-0	0.56\\
0.32	-0	0.8\\
0.64	-0	0.62\\
0.81	0	0.54\\
0.83	0	0.53\\
0.26	-0	0.84\\
};
 \addlegendentry{$\mathcal{RD}_{\mathrm{CEO}}^{2}$};

\end{axis}
\end{tikzpicture}